\documentclass[12pt]{article}
\usepackage{amssymb}
\usepackage{amsmath}
\usepackage{amsthm}
\usepackage{color}
\usepackage{comment}
\usepackage{enumitem}		
\usepackage{geometry}		
\usepackage{graphicx}
\usepackage{authblk}
\usepackage{amssymb}
\usepackage{listings}
\usepackage{xcolor}
\usepackage{caption}
\usepackage{color}

\usepackage{subcaption}

\makeatletter
	\@addtoreset{equation}{section}
\makeatother

\let\originalleft\left
\let\originalright\right
\renewcommand{\left}{\mathopen{}\mathclose\bgroup\originalleft}
\renewcommand{\right}{\aftergroup\egroup\originalright}

\newlist{romanlist}{enumerate}{3}
\setlist[romanlist]{label=\roman*),ref=(\roman*)}

\newcommand{\myStep}[2]{{\bf Step #1} --- #2\\}

\begin{document}

\newcommand{\cF}{\mathcal{F}}
\newcommand{\cR}{\mathcal{R}}
\newcommand{\cS}{\mathcal{S}}
\newcommand{\cT}{\mathcal{T}}
\newcommand{\cW}{\mathcal{W}}
\newcommand{\ep}{\varepsilon}
\newcommand{\rD}{{\rm D}}
\newcommand{\re}{{\rm e}}

\newcommand{\removableFootnote}[1]{\footnote{#1}}

\newtheorem{theorem}{Theorem}[section]
\newtheorem{corollary}[theorem]{Corollary}
\newtheorem{lemma}[theorem]{Lemma}
\newtheorem{proposition}[theorem]{Proposition}

\theoremstyle{definition}
\newtheorem{definition}{Definition}[section]
\newtheorem{example}[definition]{Example}

\theoremstyle{remark}
\newtheorem{remark}{Remark}[section]


\title{
Robust chaos in orientation-reversing and non-invertible two-dimensional piecewise-linear maps.
}

\author[1]{I.~Ghosh}
\author[1]{R.~McLachlan}
\author[1]{D.~Simpson}

\affil[1]{School of Mathematical and Computational Sciences\\
Massey University\\
Colombo Road, Palmerston North, 4410\\
New Zealand}

\maketitle


\begin{abstract}
This paper concerns the two-dimensional border-collision normal form --- a four-parameter family of piecewise-linear maps generalising the Lozi family and relevant to diverse applications. The normal form was recently shown to exhibit a chaotic attractor throughout an open region of parameter space. This was achieved by constructing a trapping region in phase space and an invariant expanding cone in tangent space, but only allowed parameter combinations for which the normal form is invertible and orientation-preserving. This paper generalises the construction to include the non-invertible and orientation-reversing cases. This provides a more complete and unified picture of robust chaos by revealing its presence to be disassociated from the global topological properties of the map. We identify a region of parameter space in which the map exhibits robust chaos, and show that part of the boundary of this region consists of bifurcation points at which the chaotic attractor is destroyed. 
\end{abstract}

\section{Introduction}
\label{sect:intro}

Robust chaos refers to the persistence of a chaotic attractor throughout open regions of the parameter space of a family of dynamical systems \cite{Gl17}.  Such robustness is essential for chaos-based cryptography \cite{KoLi11} and can be helpful to devices that have advantageous functional characteristics when operated in a chaotic regime, examples include power converters \cite{DeHa96}, optical resonators \cite{LiDi13}, and energy harvesters \cite{KuAl16}.

Robust chaos occurs in diverse settings \cite{GoGo18, ZeSp12}. Its occurrence in piecewise-linear maps was popularised by Banerjee et al.~\cite{BaYo98}, and later Glendinning and Simpson \cite{GlSi21} proved the presence of robust chaos in their setting. They provided {\em explicit} parameter values yielding robust chaos, compared to earlier results giving {\em implicit} conditions in more abstract settings \cite{Pe92,Ry04,Yo85}, but only dealt with two-dimensional, orientation-preserving maps. Piecewise-linear maps that are not orientation-preserving are arguably more physically relevant than those that are orientation-preserving, as discussed below. The purpose of this paper is to extend the results of \cite{GlSi21}
to the non-orientation-preserving cases.  This is also a necessary first step towards generalising the results to higher dimensions.

Specifically, we study the family

\begin{equation}
f_\xi(x,y) = \begin{cases}
\big( \tau_L x + y + 1, -\delta_L x \big), & x \le 0, \\
\big( \tau_R x + y + 1, -\delta_R x \big), & x \ge 0,
\end{cases}
\label{eq:BCNF2}
\end{equation}
known as the two-dimensional {\em border-collision normal form}.
For any fixed parameter point $\xi = (\tau_L,\delta_L,\tau_R,\delta_R) \in \mathbb{R}^4$,
we are interested in the behaviour of iterations $(x,y) \mapsto f_\xi(x,y)$ on $\mathbb{R}^2$.
The family \eqref{eq:BCNF2} is continuous but non-differentiable on the line $x = 0$, termed the {\em switching manifold},
and is a normal form in the sense that any two-dimensional, piecewise-linear, continuous map
with a single switching manifold and satisfying a genericity condition can be transformed to \eqref{eq:BCNF2} through a change of coordinates \cite[Appendix A]{Si22e}.
It was originally formulated by Nusse and Yorke \cite{NuYo92}
with an additional parameter $\mu$ that when varied through zero brings about a border-collision bifurcation \cite{BaGr99}.
The simplicity of \eqref{eq:BCNF2} belies the incredible complexity of its dynamics;
it can exhibit two-dimensional attractors \cite{Gl16e, GlWo11},
large numbers of coexisting attractors \cite{PuRo18, Si14, Si22c},
and Arnold tongues with a unique sausage-string geometry \cite{Si17c}.
In the special case $\tau_L = -\tau_R$ and $\delta_L = \delta_R$,
\eqref{eq:BCNF2} reduces to the Lozi family \cite{Lo78}.

At any point with $x \ne 0$ the derivative of \eqref{eq:BCNF2} has determinant
\begin{equation}
\det \big( ({\rm D} f_\xi)(x,y) \big) =
\begin{cases}
\delta_L, & x < 0, \\
\delta_R, & x > 0.
\end{cases}
\nonumber
\end{equation}
Consequently we have the following cases.
\begin{itemize}
\setlength{\itemsep}{0pt}
\item
If $\delta_L > 0$ and $\delta_R > 0$ then $f_\xi$ is orientation-preserving.
This is the most natural case to consider from an applied perspective
because Poincar\'e maps derived from time-reversible flows on $\mathbb{R}^n$ are necessarily orientation-preserving \cite{Wi88}.
\item
If $\delta_L < 0$ and $\delta_R < 0$ then $f_\xi$ is orientation-reversing.
This case was considered in the seminal work of Misiurewicz \cite{Mi80} for the Lozi map. Such maps can be embedded into higher dimensional orientation-preserving maps and in this way help us understand high-dimensional chaos in physical systems \cite{KaSh21}.
\item
If $\delta_L$ and $\delta_R$ have opposite signs then $f_\xi$ is non-invertible
with the left and right half-planes mapping onto either the top or the bottom half-plane.
Such maps, except piecewise-smooth instead of piecewise-linear,
apply to a wide range of power converters \cite{BaVe01,DiTs02,ZhMo03}.
For instance, the boost converter described by Banerjee et al.~\cite{BaYo98}
regulates output voltage via a switch that is activated whenever the current reaches a threshold value.
This causes the flow (in phase space) to fold back over itself resulting in a stroboscopic map that
is piecewise-smooth (due to the switch) and non-invertible (due to the folding) \cite{De92, FaSi22}.
\item
If one of $\delta_L$ and $\delta_R$ is zero then $f_\xi$ is non-invertible with one half-plane mapping onto a line \cite{Ko04}.
Such maps arise as leading-order approximations to grazing-sliding bifurcations in relay control systems \cite{DiKo02} and mechanical systems with stick-slip friction \cite{DiKo03,SzOs08}. This is because trajectories of the differential equations become constrained to a codimension-one discontinuity surface
causing the range of one piece of the Poincar\'e map to have one less dimension than the domain of the map \cite{DiKo02}. This also occurs for the Hickian trade cycle model of Puu et al.~\cite{PuSu06} and the influenza outbreak model of Roberts et al.~\cite{RoHi19b}.
\end{itemize}

Let
\begin{equation}
\Phi = \left\{ \xi \in \mathbb{R}^4 \,\middle|\, \tau_L > |\delta_L + 1|,\, \tau_R < |\delta_R + 1| \right\}
\label{eq:Phi}
\end{equation}
denote the set of all $\xi$ for which $f_\xi$ has two saddle fixed points (Lemma \ref{le:Phi}).
Banerjee {\em et al.}~\cite{BaYo98} considered $\xi \in \Phi$ in the orientation-preserving case
and showed that on one side of a homoclinic bifurcation
the stable manifold of one fixed point has transverse intersections with the unstable manifolds of both fixed points
and argued that this implies the existence of a chaotic attractor.
Their conclusion was verified rigorously by Glendinning and Simpson~\cite{GlSi21} using the methodology of Misiurewicz \cite{Mi80}.
First, a forward invariant region $\Omega \subset \mathbb{R}^2$ was identified by using one of the unstable manifolds,
and this was perturbed into a {\em trapping region} that necessarily contains a topological attractor.
Second, an {\em invariant expanding cone} $\Psi \subset T \mathbb{R}^2$ (see \S\ref{sec:cone})
was constructed by using eigenvectors of the two pieces of ${\rm D} f_\xi$.
The existence of this object implies that nearby forward orbits diverge exponentially
and that the attractor is chaotic in the sense of having a positive Lyapunov exponent. 
In general this approach works brilliantly for piecewise-linear maps because the invariance and expansion properties can be verified by simple, explicit computations.
For more complicated constructions the verification is best done on a computer \cite{GlSi22b}.
The construction is robust to nonlinear perturbations to the pieces of the map
and has been used to show that chaos persists for intervals of parameter values beyond border-collision bifurcations \cite{SiGl21}.
Recently this approach has also been applied to piecewise-smooth maps with a square-root singularity that arise for mechanical systems with hard impacts \cite{MiLi19}.
Further techniques can reveal more properties of the attractor,
such as sensitive dependence on initial conditions \cite{GhSi22b},
continuity with respect to $\xi$ \cite{AlPu17,Ku22},
and the presence of an SRB measure \cite{Gl17}.

In this paper, we extend the construction of \cite{GlSi21} to the orientation-reversing and non-invertible cases.
We obtain a subset $\Phi_{\rm trap} \subset \Phi$ in which $f_\xi$ has a trapping region
and another subset $\Phi_{\rm cone} \subset \Phi$ in which $f_\xi$ has an invariant expanding cone.
It follows that $f_\xi$ has a chaotic attractor for all $\xi \in \Phi_{\rm trap} \cap \Phi_{\rm cone}$.
This is an open subset of parameter space, hence the chaos is robust with respect to the family $f_\xi$. We expect it is also robust to nonlinear perturbations to the pieces of the map, as in \cite{SiGl21}.
Boundaries of $\Phi_{\rm trap} \cap \Phi_{\rm cone}$ are where some aspect of the construction fails.
As shown below, three of these boundaries correspond to bifurcations where the chaotic attractor is destroyed.
Beyond the other boundaries the chaotic attractor appears to persist and we believe robust chaos could be verified on a larger subset of parameter space by using a more complicated construction, e.g.~\cite{Si22e},
but our aim here is not to optimise the subset, only to obtain a reasonably large subset
for both negative and positive values of $\delta_L$ and $\delta_R$ by using a construction that is both natural and simple.

The remainder of this paper is organised as follows.
We start in \S\ref{sec:main} by calculating the 
stable and unstable manifolds of the fixed points.
Constraints on the geometry of the manifolds as they are extended outwards from the fixed points give rise to our definition of $\Phi_{\rm trap}$.
Here we also define $\Phi_{\rm cone}$
and state our main result (Theorem \ref{thm:attractor}) for the existence of a chaotic attractor.
Being four-dimensional the sets $\Phi_{\rm trap}$ and $\Phi_{\rm cone}$ are difficult to visualise;
we show a range of two-dimensional cross-sections to give some impressions of their size and shape.

In \S\ref{sec:for_inv_trap} we use the stable manifold of one of the fixed points
to form a region $\Omega \subset \mathbb{R}^2$ and show it is forward invariant under $f_\xi$ for any $\xi \in \Phi_{\rm trap}$.
We then show how $\Omega$ can be perturbed into a trapping region for any $\xi \in \Phi_{\rm trap}$.
In \S\ref{sec:cone} we identify a cone that is both invariant and expanding for any $\xi \in \Phi_{\rm cone}$
and prove Theorem \ref{thm:attractor}.
In \S\ref{sec:topologies} we further study the extent to which $\Phi_{\rm trap}$ and $\Phi_{\rm cone}$ cover parameter space.  We show a typical cross-section of $\Phi_{\rm trap} \cap \Phi_{\rm cone}$ and overlay numerical simulations whereby the presence robust chaos is estimated from forward orbits.  This helps clarify which boundaries of $\Phi_{\rm trap}$ and $\Phi_{\rm cone}$ correspond to bifurcations and how much of the true robust chaos region is detected by the straight-forward constructions.  Final remarks are provided in \S\ref{sec:dis}.

\section{Sufficient conditions for a chaotic attractor}
\label{sec:main}

For any $\xi \in \Phi$ the map $f_\xi$ has fixed points
\begin{align}
X &= \left(\frac{-1}{\tau_R-\delta_R-1}, \frac{\delta_R}{\tau_R-\delta_R-1}\right),
\label{eq:X} \\
Y &= \left(\frac{-1}{\tau_L-\delta_L-1}, \frac{\delta_L}{\tau_L-\delta_L-1}\right),
\label{eq:Y}
\end{align}
where $X$ is in the right half-plane and $Y$ is in the left half-plane.
The stability multipliers associated with $X$ and $Y$
are the eigenvalues of the Jacobian matrices
\begin{align}
A_L &= \begin{bmatrix} \tau_L & 1 \\ -\delta_L & 0 \end{bmatrix}, &
A_R &= \begin{bmatrix} \tau_R & 1 \\ -\delta_R & 0 \end{bmatrix}.
\label{eq:ALAR}
\end{align}
Since $\xi \in \Phi$, $A_L$ has eigenvalues $\lambda_L^s \in (-1,1)$ and $\lambda_L^u>1$,
while $A_R$ has eigenvalues $\lambda_R^s \in (-1,1)$ and $\lambda_R^u<-1$.
This implies $X$ and $Y$ are saddles.
In fact, $\Phi$ is the set of all parameter combinations for which
$f_\xi$ has two saddle fixed points:

\begin{lemma}
The map $f_\xi$ has two saddle fixed points if and only if $\xi \in \Phi$.
\label{le:Phi}
\end{lemma}


\begin{proof}
Let $f^L$ and $f^R$ denote the left and right pieces of $f_\xi$, respectively.
If $\delta_L = \tau_L - 1$ then $f^L$ has no fixed points,
while if $\delta_L \ne \tau_L - 1$ then $f^L$ has the unique fixed point $Y$, given by \eqref{eq:Y}.
Similarly if $\delta_R = \tau_R - 1$ then $f^R$ has no fixed points,
while if $\delta_R \ne \tau_R - 1$ then $f^R$ has the unique fixed point $X$, given by \eqref{eq:X}.

If $\xi \in \Phi$ then $X$ and $Y$ are saddle fixed points of $f_\xi$, as noted above.
Conversely, suppose $f_\xi$ has two saddle fixed points.
By the above remarks, these must be $X$ and $Y$.
Since they are fixed points of $f_\xi$,
$Y$ is in the left half-plane, so $\tau_L > \delta_L + 1$ by \eqref{eq:Y},
and $X$ is in the right half-plane, so $\tau_R < \delta_R + 1$ by \eqref{eq:X}.
Since they are saddles, also $\tau_L > -(\delta_L + 1)$
and $\tau_R < -(\delta_R + 1)$, hence $\xi \in \Phi$.
\end{proof}

\subsection{The stable and unstable manifolds of the fixed points}
\label{sec:stable_unstable_manifolds}

Since $X$ and $Y$ are saddles they have one-dimensional stable and unstable manifolds.
Their stable manifolds $W^s(X)$ and $W^s(Y)$
have kinks at points where they meet the switching manifold $x = 0$
and on preimages of these points,
while their unstable manifolds $W^u(X)$ and $W^u(Y)$
have kinks at points where they meet $y = 0$
(the image of the switching manifold) and on images of these points.
As each manifold (stable or unstable) emanates from the fixed point,
it coincides with the line through the fixed point with direction
given by the corresponding eigenvector of $A_L$ or $A_R$.
We will use a subscript $0$ to denote
the part of the manifold that coincides with this line.
These are indicated in Fig.~\ref{fig:typ_pp} for four different combinations of the parameter values.

\begin{figure}[h]
\begin{center}
\begin{tabular}{cc}
  \includegraphics[scale=0.48]{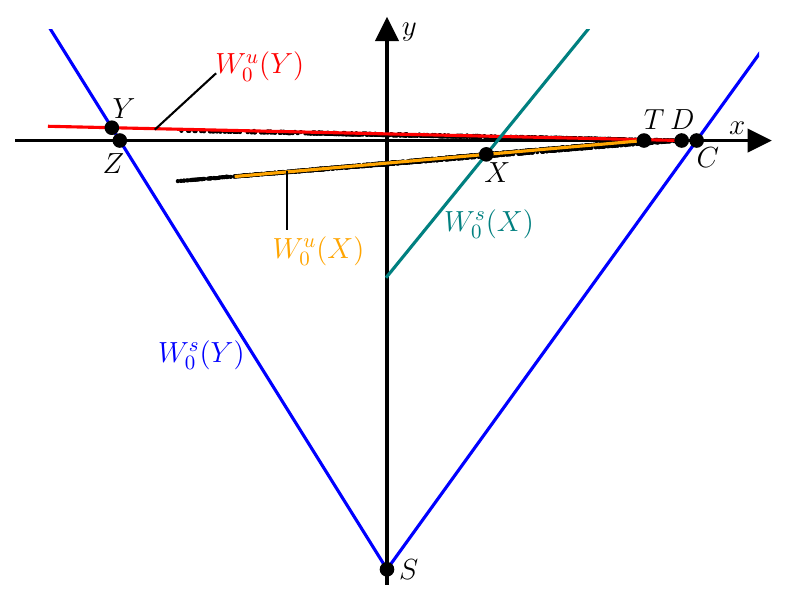} &  \includegraphics[scale=0.48]{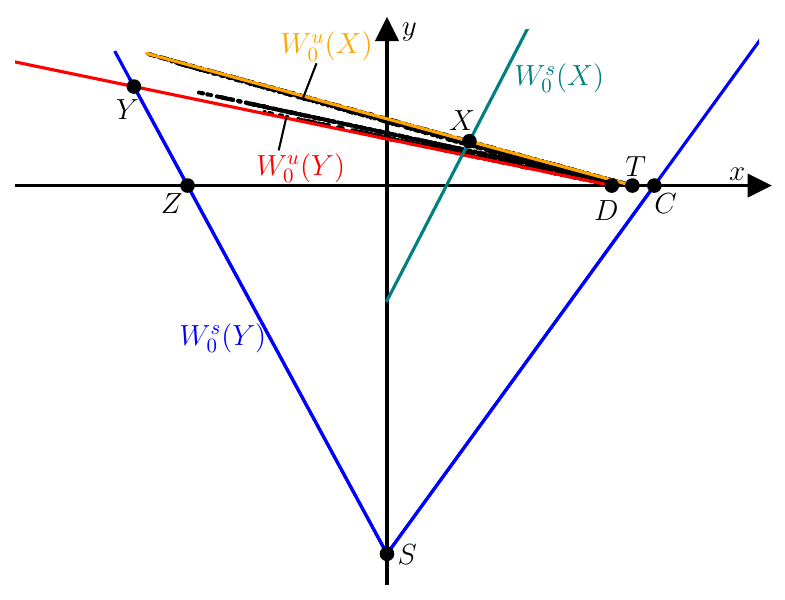} \\
(a) $\delta_L>0, \delta_R>0$ & (b) $\delta_L>0, \delta_R<0$ \\[3pt]
 \includegraphics[scale=0.48]{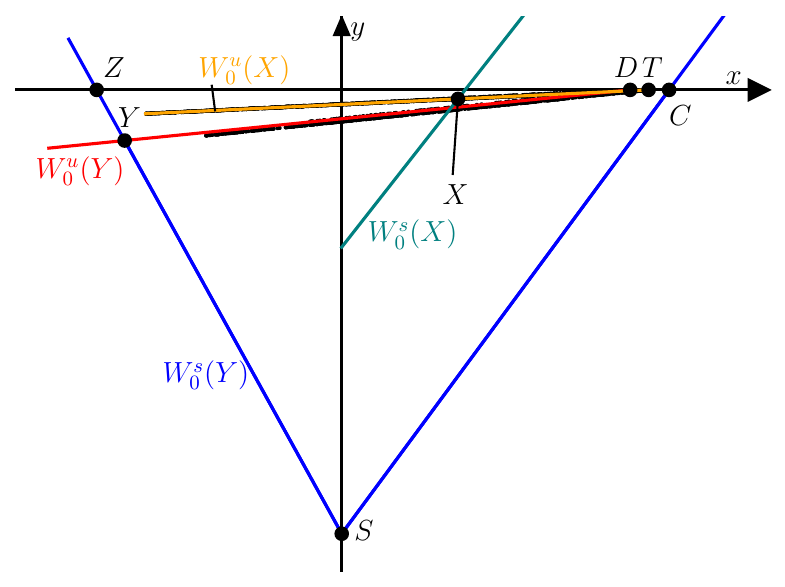} &   \includegraphics[scale=0.48]{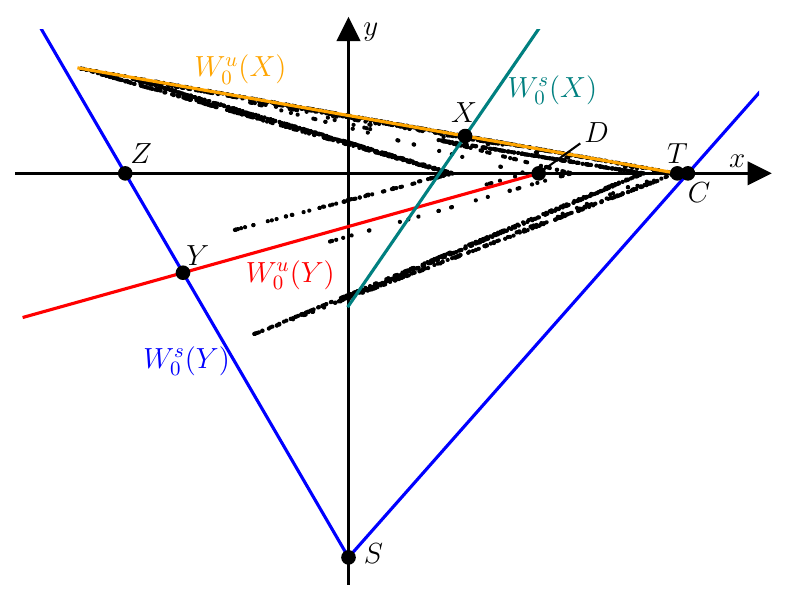} \\
(c) $\delta_L<0, \delta_R>0$ & (d) $\delta_L<0, \delta_R<0$ \\[3pt]
\end{tabular}
\end{center}
\caption{
Phase portraits of $f_\xi$ \eqref{eq:BCNF2} for four different parameter
combinations $\xi = (\tau_L,\delta_L,\tau_R,\delta_R) \in \Phi_{\rm trap} \cap \Phi_{\rm cone}$; (a) $\xi = (2.1, 0.06, -1.7, 0.18)$; (b) $\xi = (2.1, 0.4, -1.7, -0.55)$; (c) $\xi = (2.2, -0.3, -1.7, 0.1)$; (d) $\xi = (1.8, -0.75, -1.6, -0.4)$.
Each plot shows the initial line segments of the stable and unstable manifolds emanating from
the fixed points $X$ and $Y$, as well as an additional segment of the stable manifold
of $Y$. We also show some key points of intersection ($S$, $C$, $D$, $T$, and $Z$) of these segments 
with the coordinate axes, see \eqref{eq:S}, \eqref{eq:C}, \eqref{eq:D}, \eqref{eq:T}, and \eqref{eq:Z} respectively.
To illustrate the chaotic attractor the black dots show iterates of a
typical forward orbit after transient dynamics has decayed.
}
\label{fig:typ_pp}
\end{figure}


Next, we describe some points where the stable and unstable manifolds
intersect $x=0$ and $y=0$ as these are central to our construction in \S\ref{sec:for_inv_trap}.
For all $\xi \in \Phi$, $W^s_0(Y)$ has an endpoint on $x=0$, call it $S$.
Since $W^s_0(Y)$ has slope $-\lambda_L^u$ (due to the companion matrix form \eqref{eq:ALAR}),
from the above formula for $Y$ we obtain
\begin{equation}
\label{eq:S}
S = \left(0, \frac{-\lambda_L^u}{\lambda_L^u - 1} \right).
\end{equation}
Notice $S_2 < -1$
(here and throughout the paper for any $P \in \mathbb{R}^2$
we write $P_1$ and $P_2$, respectively,
for its $x$ and $y$ components).
The manifold $W^s(Y)$ continues into the right half-plane
but now with slope $\frac{\phi_1(\xi)}{\lambda_L^u}$, where
\begin{equation}
\label{eq:phi1}
\phi_1(\xi) = \delta_R - \tau_R\lambda_L^u \,.
\end{equation}
Our trapping region construction requires that this
linear segment, call it $W_1^s(Y)$, intersects $y=0$.
Certainly this is only possible if $\phi_1(\xi) > 0$.
This inequality turns out to be sufficient to ensure
$W_1^s(Y)$ intersects $y=0$ except in the case $\delta_L, \delta_R < 0$ (Fig.~\ref{fig:typ_pp}-d).
In this case $W_1^s(Y)$ is the line segment
connecting $S$ and $f_\xi^{-2}(S)$.
So in this case for $W_1^s(Y)$ to intersect $y = 0$ we need $f_\xi^{-2}(S)$ to lie above $y=0$.
A straightforward calculation reveals that the $y$-component of $f_\xi^{-2}(S)$
is $\frac{\phi_2(\xi)}{\lambda_L^s (\lambda_L^u - 1) \delta_R}$, where
\begin{equation}
\label{eq:phi2}
\phi_2(\xi) = \delta_R(\lambda_L^s+1) - \lambda_L^u(\tau_R + (\delta_R + \tau_R)\lambda_L^s).
\end{equation}
Thus we require $\phi_2(\xi) > 0$ (because with $\delta_L, \delta_R < 0$ we have $\lambda_L^s (\lambda_L^u - 1) \delta_R > 0$).
In any case, if $W_1^s(Y)$ intersects $y = 0$,
it does so at the point
\begin{equation}
\label{eq:C}
C = \left( 1 + \frac{\phi_3(\xi)}{(\lambda_L^u - 1) \phi_1(\xi)}, 0 \right),
\end{equation}
where
\begin{equation}
\label{eq:phi3}
\phi_3(\xi) = \delta_R - (\delta_R +\tau_R - (\tau_R+1)\lambda_L^u)\lambda_L^u.
\end{equation}
Our construction also requires $C_1 > 1$, that is $\phi_3(\xi) > 0$.

Now we consider the two unstable manifolds.
For all $\xi \in \Phi$, $W^u_0(Y)$ has an endpoint on $y=0$ at
\begin{equation}
\label{eq:D}
D = \left(\frac{1}{1-\lambda_L^s}, 0 \right),
\end{equation}
and $W^u_0(X)$ has an endpoint on $y=0$ at
\begin{equation}
\label{eq:T}
T = \left(\frac{1}{1-\lambda_R^s} , 0\right).
\end{equation}
These points are indicated in Fig.~\ref{fig:typ_pp}.

\subsection{Homoclinic and heteroclinic bifurcations}

In the orientation-preserving case ($\delta_L, \delta_R > 0$),
Banerjee et al.~\cite{BaYo98} noticed that as parameters are varied
a chaotic attractor can be destroyed when the points $C$ and $D$ collide.
This type of bifurcation is a {\em homoclinic corner} \cite{Si16b}
where the kinks (corners) of the unstable manifold of $Y$ lie on the stable manifold of $Y$ (and vice-versa).
This is analogous to a first homoclinic tangency \cite{PaTa93} for smooth maps
which are well understood as a mechanism for the destruction of an attractor \cite{GrOt83}.
From the above formulas for $C$ and $D$ we obtain
\begin{align}
\label{eq:phi4origin}
C_1 - D_1 = \frac{\phi_4(\xi)}{(\lambda_L^u - 1)(1 - \lambda_L^s)\phi_1(\xi)},
\end{align}
where
\begin{align}
\label{eq:phi4}
\phi_4(\xi) = \delta_R - (\tau_R+\delta_L+\delta_R - (1+\tau_R)\lambda_L^u)\lambda_L^u.
\end{align}
So the condition $\phi_4(\xi) > 0$ (equivalent to equation (5) of \cite{BaYo98}) ensures that $C$ lies to the right of $D$ as in Fig.~\ref{fig:typ_pp}-a.

In the orientation-reversing case ($\delta_L, \delta_R < 0$)
a chaotic attractor can be destroyed when the points $C$ and $T$ collide.
Here kinks of the unstable manifold of $X$
lie on the stable manifold of $Y$ (and vice-versa).
We have
\begin{align}
\label{eq:T1C1}
C_1 - T_1 = \frac{\phi_5(\xi)}{(\lambda_L^u-1)(1-\lambda_R^s)\phi_1(\xi)},
\end{align}
where
\begin{align}
\label{eq:phi5}
\phi_5(\xi) = \delta_R - (\delta_R + \tau_R - (1 +\lambda_R^u)\lambda_L^u)\lambda_L^u.
\end{align}
The condition $\phi_5(\xi) > 0$ ensures that $C$ lies to the right of $T$
as in Fig.~\ref{fig:typ_pp}-d.
In the special case of the Lozi map, $\phi_5(\xi) > 0$ simplifies
(significantly) to equation (3) of Misiurewicz \cite{Mi80}.

In view of the above discussion we define
\begin{align}
\label{eq:Phi_trap}
\Phi_{\rm trap} = \left\{ \xi \in \Phi \,\middle|\,  \phi_i(\xi) > 0, i=1, \ldots, 5 \right\}.
\end{align}
This set is difficult to visualise because parameter space is four-dimensional.
Fig.~\ref{fig:phis} shows four different cross-sections of $\Phi_{\rm trap}$ obtained by fixing $\tau_L$ and $\tau_R$.
Broadly speaking the size of the cross-section decreases as the values of $\tau_L$ and $|\tau_R|$ increase.
Notice the topology of the cross-sections is different for different values of $\tau_L$ and $\tau_R$.
For instance in Fig.~\ref{fig:phis}-a the boundary of the cross-section is formed by $\phi_1(\xi) = 0$, $\phi_2(\xi) = 0$,
and the boundary of $\Phi$,
whereas in Fig.~\ref{fig:phis}-d the boundary is formed by $\phi_1(\xi) = 0$, $\phi_3(\xi) = 0$, $\phi_4(\xi) = 0$, $\phi_5(\xi) = 0$,
and the boundary of $\Phi$.
The figure also includes curves on which $f_\xi(C) = Y$ and $f_\xi(C) = Z$,
where $Z = f_\xi(S)$ is the intersection of $W_0^s(Y)$ with $y=0$.
Together with the $\delta_L$ and $\delta_R$ axes, these curves divide the cross-sections into six parts
corresponding to six cases for the vertices of the region $\Omega$
that we construct in \S\ref{sec:for_inv_trap}.

\begin{figure}[h]
\begin{center}
\begin{tabular}{cc}
 \includegraphics[scale=0.41]{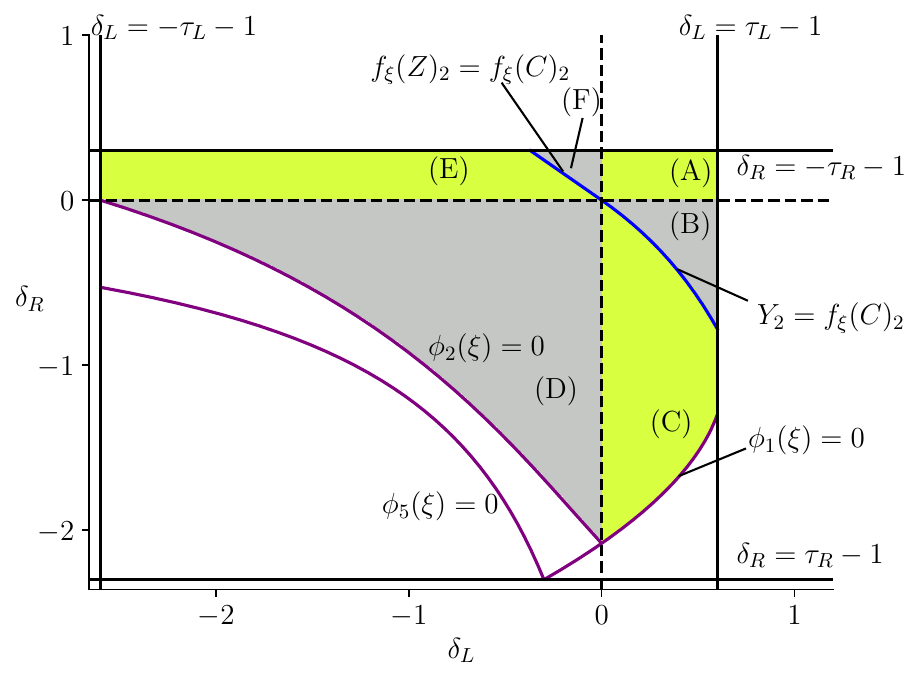} &   \includegraphics[scale=0.41]{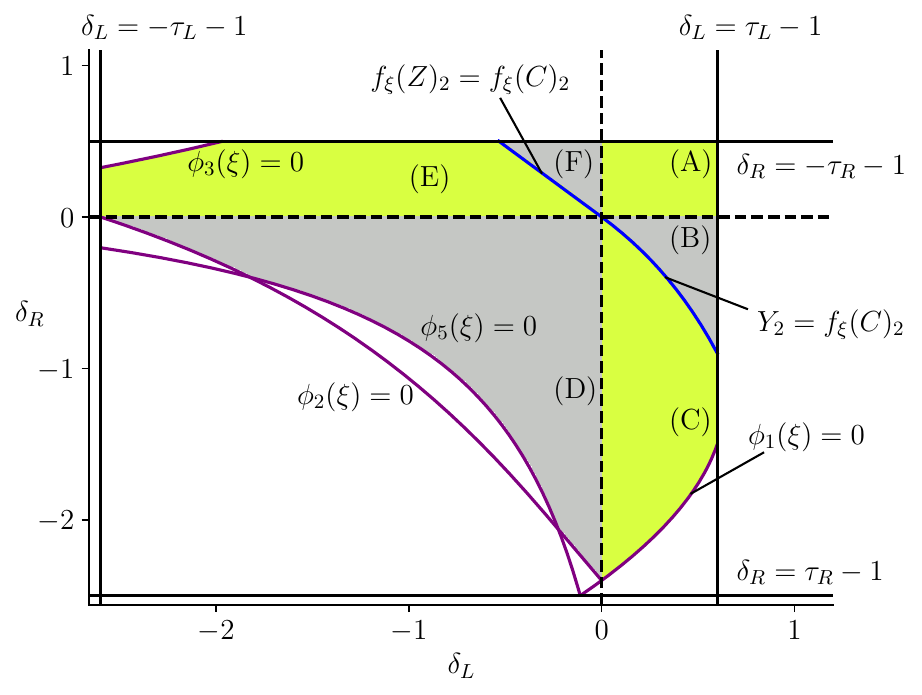} \\
(a) $\tau_L=1.6, \tau_R=-1.3$ & (b) $\tau_L=1.6, \tau_R=-1.5$ \\[6pt]
\includegraphics[scale=0.41]{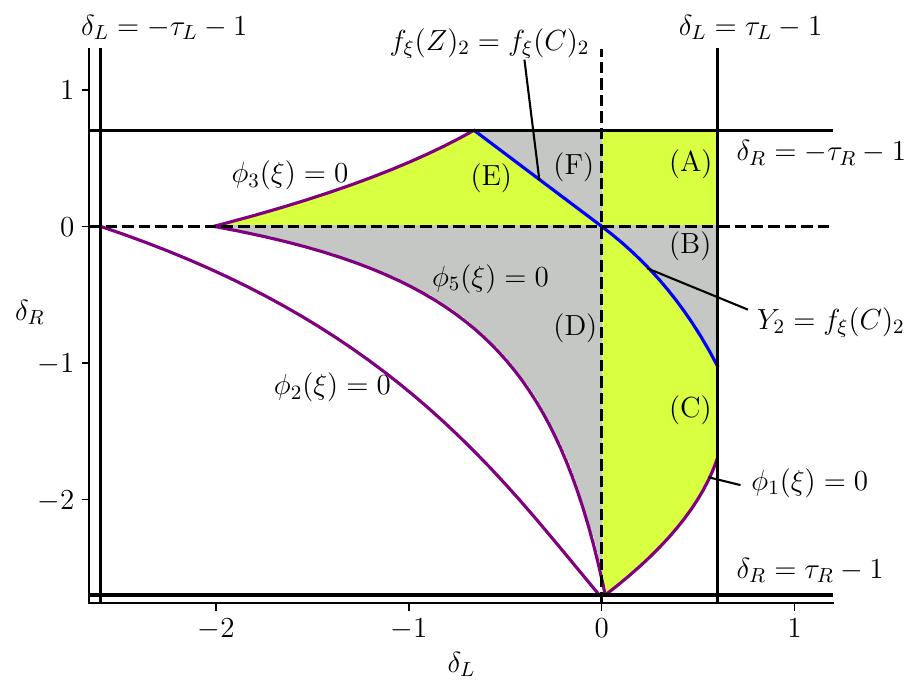} &   \includegraphics[scale=0.41]{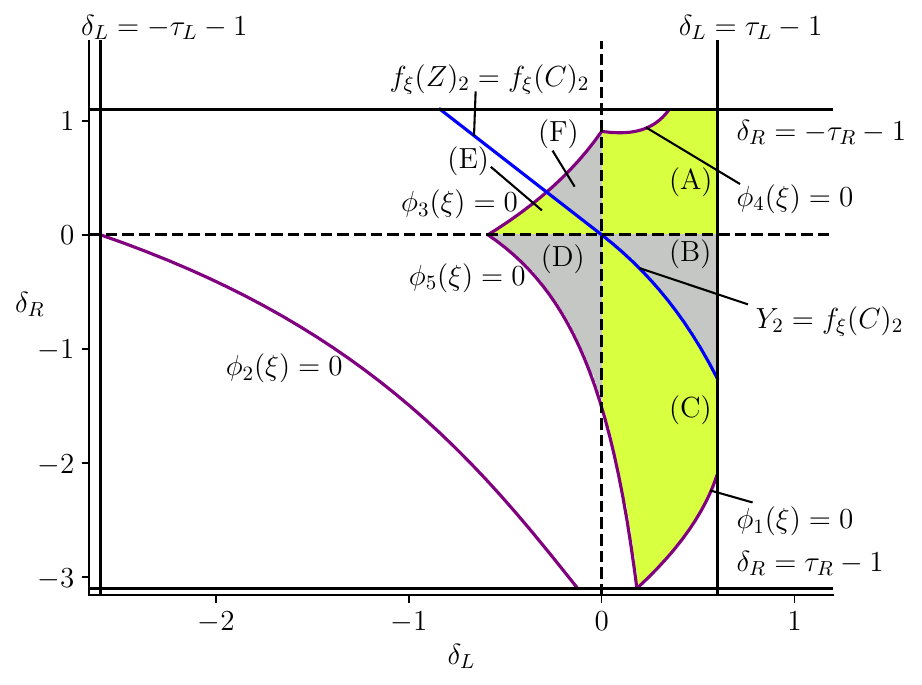} \\
(c) $\tau_L=1.6, \tau_R=-1.7$ & (d) $\tau_L=1.6, \tau_R=-2.1$ \\[6pt]
\end{tabular}
\end{center}
\caption{
Four different two-dimensional cross sections of $\Phi_{\rm trap}$,
defined by \eqref{eq:Phi_trap}.
In each plot the cross-section is divided into six parts
(A)--(F) where $\Omega$ has certain vertices, see \S\ref{sec:for_inv_trap}.
}
\label{fig:phis}
\end{figure}

\subsection{Sufficient conditions for robust chaos}

Our construction of an invariant expanding cone $\Psi_K$
requires similar constraints on the parameter values to those established above for the trapping region. To this end we define
\begin{align}
\theta_1(\xi) &= \left( \delta_L + \delta_R - \tau_L \tau_R \right)^2
- 4 \delta_L \delta_R \,, \label{eq:theta1} \\
\theta_2(\xi) &= \tau_L^2 + \delta_L^2 - 1 + 2 \tau_L \,{\rm min}\left( 0, -\tfrac{\delta_R}{\tau_R}, q_L, \tilde{a} \right), \label{eq:theta2} \\
\theta_3(\xi) &= \tau_R^2 + \delta_R^2 - 1 + 2 \tau_R \,{\rm max}\left( 0,
-\tfrac{\delta_L}{\tau_L}, q_R, \tilde{b} \right), \label{eq:theta3}
\end{align}
where
\begin{align}
q_L &= -\tfrac{\tau_L}{2}\left(1 - \sqrt{1 - \tfrac{4\delta_L}{\tau_L^2}} \right), &
q_R &= -\tfrac{\tau_R}{2}\left(1 - \sqrt{1 - \tfrac{4\delta_R}{\tau_R^2}} \right), \nonumber
\end{align}
and
\begin{align}
\tilde{a} &= \frac{\delta_L -\delta_R-\tau_L\tau_R - \sqrt{\theta_1(\xi)}}{2\tau_R}, &
\tilde{b} &= \frac{\delta_R -\delta_L-\tau_L\tau_R - \sqrt{\theta_1(\xi)}}{2\tau_L}, \nonumber
\end{align}
assuming $\theta_1(\xi) > 0$.
We then define
\begin{align}
\label{eq:Phi_cone}
\Phi_{\rm cone} = \left\{ \xi \in \Phi \,\middle|\, \theta_i(\xi) > 0, i=1, \ldots, 3 \right\}.
\end{align}
The condition $\theta_1(\xi) > 0$ ensures $\theta_2(\xi)$ and $\theta_3(\xi)$ are well-defined,
and, as explained in \S\ref{sec:cone}, the conditions $\theta_2(\xi) > 0$ and $\theta_3(\xi) > 0$ ensure that our cone $\Psi_K$
is invariant and expanding.

Fig.~\ref{fig:cone_regions} shows cross-sections of $\Phi_{\rm cone}$.
Broadly speaking the size of the cross-section
{\em increases} as the values of $\tau_L$ and $|\tau_R|$ increase.
Again the topology of the cross-sections is different for different values of $\tau_L$ and $\tau_R$.
Similar to Fig.~\ref{fig:phis} we have divided the cross-sections into six parts corresponding to six cases for the boundary of $\Psi_K$ (see \S\ref{sec:cone}).
These correspond to different cases for the four quantities in each of
\eqref{eq:theta2} and \eqref{eq:theta3}
that attain the minimum (respectively, maximum) value.

\begin{figure}[h]
\begin{center}
\begin{tabular}{cc}
  \includegraphics[scale=0.41]{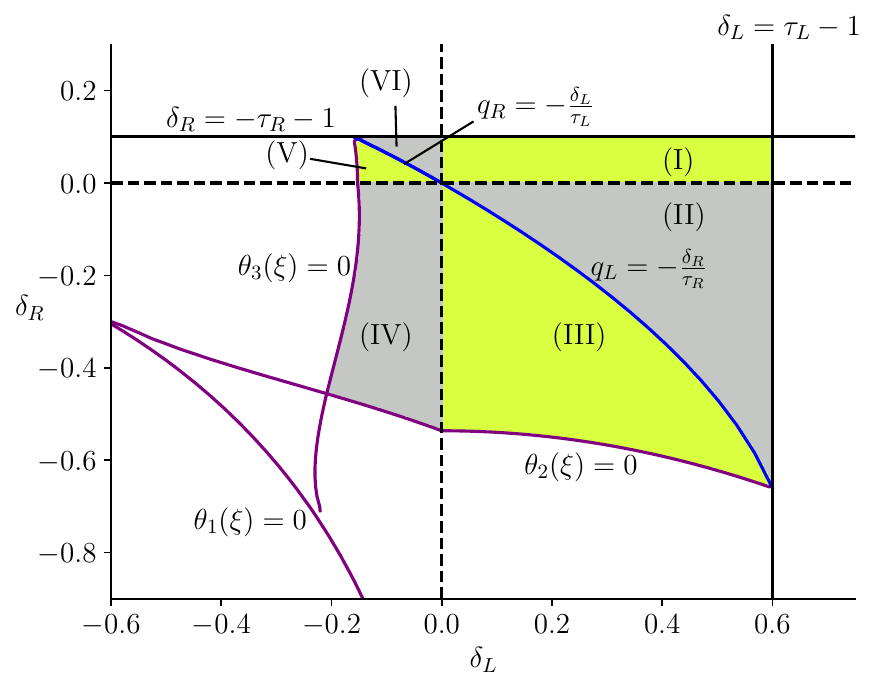} &   \includegraphics[scale=0.41]{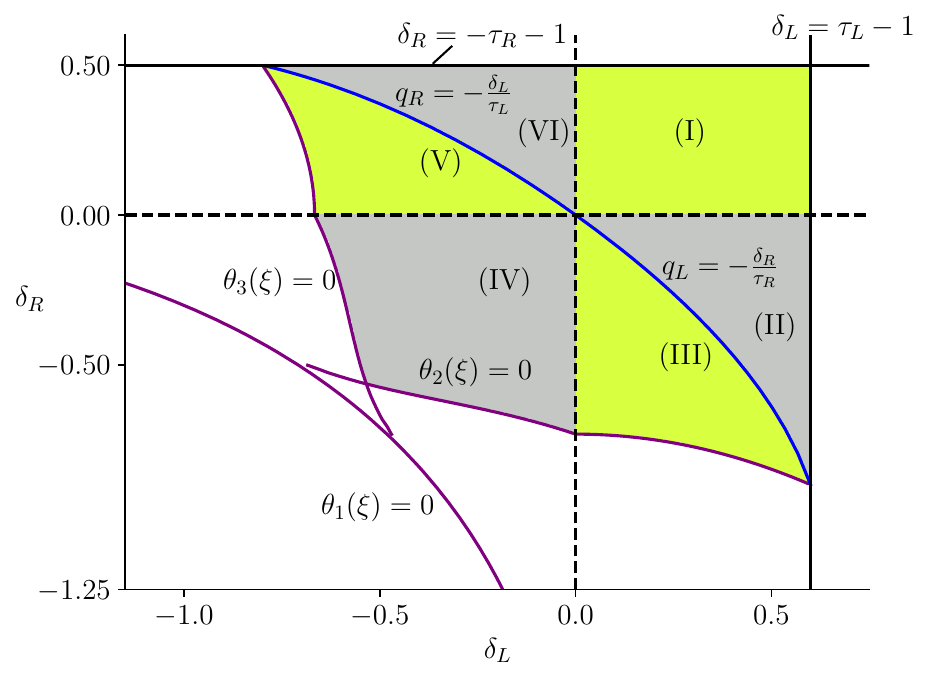} \\
(a) $\tau_L=1.6, \tau_R=-1.1$ & (b) $\tau_L=1.6, \tau_R=-1.5$ \\[3pt]
 \includegraphics[scale=0.41]{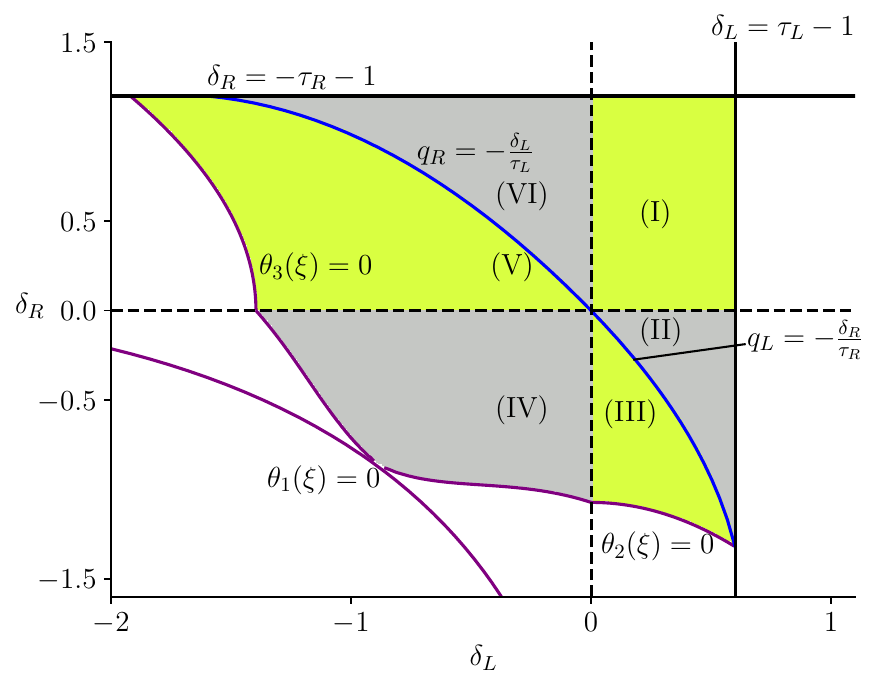} &   \includegraphics[scale=0.41]{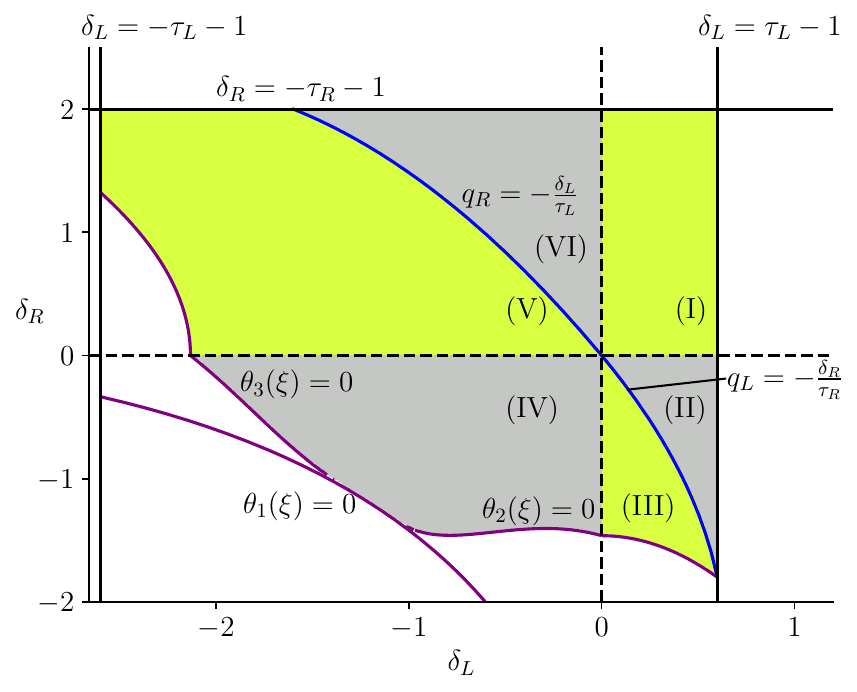} \\
(c) $\tau_L=1.6, \tau_R=-2.2$ & (d) $\tau_L=1.6, \tau_R=-3$ \\[3pt]
\end{tabular}
\end{center}
\caption{Two-dimensional cross-sections of $\Phi_{\rm cone}$, defined by~\eqref{eq:Phi_cone}. In each plot the cross-section has been divided into six parts (I)--(VI) where the cone $\Psi_K$ has different boundaries, see Fig.~\ref{fig:slope_maps}.}\label{fig:cone_regions}
\end{figure}

Finally, we can state our main result.

\begin{theorem}
\label{thm:attractor}
For any $\xi \in \Phi_{\rm trap} \cap \Phi_{\rm cone}$
the normal form $f_\xi$ \eqref{eq:BCNF2}
has a topological attractor with a positive Lyapunov exponent.
\end{theorem}

This is proved at the end of \S\ref{sec:cone}.
We have found that there are many possibilities for the topology
of cross-sections of $\Phi_{\rm trap} \cap \Phi_{\rm cone}$ defining by fixing $\tau_L$ and $\tau_R$,
and we do not attempt to catagorise these in this paper.
We provide one example in \S\ref{sec:topologies}, where we also identify critical values of $\tau_L$ and $\tau_R$
at which the cross-sections of $\Phi_{\rm trap}$ and $\Phi_{\rm cone}$ vanish entirely.

\section{A forward invariant region and a trapping region}
\label{sec:for_inv_trap}

The following definition applies to any continuous map $f$ on the plane.
Note we write ${\rm int}(\cdot)$ for the {\em interior} of a set.

\begin{definition}
A set $\Omega \subset \mathbb{R}^2$ is {\em forward invariant} if $f(\Omega) \subset \Omega$.
A compact set $\Omega \subset \mathbb{R}^2$ is a {\em trapping region} if $f(\Omega) \subset {\rm int}(\Omega)$.
\label{df:trap}
\end{definition}

In this section we construct a triangle $\Omega$
and show that for any $\xi \in \Phi_{\rm trap}$
this region is forward invariant under $f$ (Proposition \ref{pr:for_inv}).
We then show there exists a perturbation of $\Omega$
that is a trapping region for $f$ (Proposition \ref{pr:trap}).
So that the non-invertible cases can be accommodated, our $\Omega$ differs from the triangle constructed by Glendinning and Simpson~\cite{GlSi21}
for the orientation-preserving case,
and by Misiurewicz \cite{Mi80} for the orientation-reversing case.
For clarity we now suppress the $\xi$-dependency and write $f$ instead of $f_{\xi}$.

The linear segment of the stable manifold of $Y$ that contains $Y$,
denoted $W^s_0(Y)$, was shown in Fig.~\ref{fig:typ_pp}
for four different combinations of the parameter values.
In any case, this segment lies on the line
$y = -\lambda_L^u x + S_2$, where $S_2$ is the
$y$-component of $S$, given by \eqref{eq:S}.
The point $S$ is the right-most point of $W^s_0(Y)$,
and is easy to see that in the other direction $W^s_0(Y)$ extends to infinity if $\delta_L \ge 0$
and to the preimage of $S$ under the left piece of $f$ otherwise:

\begin{lemma}
\label{lem:W^s_0(Y)}
Suppose $\delta_L \in \mathbb{R}$ and $\tau_L > |\delta_L + 1|$.
Then
\begin{equation}
W^s_0(Y) = \begin{cases}
\left\{ (x,y) \,\middle|\, -\frac{S_2}{\delta_L} \le x \le 0,
y = -\lambda_L^u x + S_2 \right\}, & \delta_L < 0, \\
\left\{ (x,y) \,\middle|\, -\infty < x \le 0,
y = -\lambda_L^u x + S_2 \right\}, & \delta_L \ge 0,
\end{cases}
\end{equation}
and $f(W^s_0(Y)) \subset W^s_0(Y)$.
\end{lemma}

In particular $Z = f(S) \in W^s_0(Y)$.
This point is the intersection of $W^s_0(Y)$ with $y=0$ and given by
\begin{align}
\label{eq:Z}
Z &= \left(\frac{-1}{\lambda_L^u-1}, 0 \right).
\end{align}
Now recall if $\phi_1(\xi) > 0$ and $\phi_2(\xi) > 0$
then $C \in W_1^s(Y)$ is given by \eqref{eq:C}.

\begin{lemma}
\label{lem:points}
Let $\xi \in \Phi$ with $\phi_1(\xi) > 0$ and $\phi_2(\xi) > 0$.
Then
\begin{equation}
Y, Z, f(Z), f(C), f^2(C) \in W^s_0(Y) \setminus \{ S \}.
\end{equation}
\end{lemma}

\begin{proof}
Certainly $Y, Z \in W^s_0(Y) \setminus \{ S \}$ by construction.
Also $f(Z) \in W^s_0(Y) \setminus \{ S \}$
because $W^s_0(Y)$ is forward invariant (Lemma \ref{lem:W^s_0(Y)})
and $Z = f(S)$ cannot be a preimage of $S$.

Since $C_1 >0$ we have $f(C) = \left(\tau_R C_1+1, -\delta_RC_1\right)$,
and it is a simple exercise to use the formula \eqref{eq:D}
to show that $f(C)$ lies on the line $y = -\lambda_L^u x + S_2$.
Also $f(C)$ lies to the left of $S$ because
$f(C)_1 = \frac{(\tau_R + \delta_R)(\lambda_L^u - 1) + \tau_R}{(\lambda_L^u - 1)\phi_1(\xi)}$ is negative by inspection,
and in the case $\delta_L < 0$ the point $f(C)$ lies to the right
of the left-most point of $W_0^s(Y)$ because
$f(C)_1 + \frac{S_2}{\delta_L} = \frac{\phi_2(\xi)}{-\lambda_L^s (\lambda_L^u - 1) \phi_1(\xi)}$ is positive.
Thus $f(C)$ belongs to $W_0^s(Y)$ and is not an endpoint of $W_0^s(Y)$,
thus $f(C) \in W^s_0(Y) \setminus \{ S \}$
and $f^2(C) \in W^s_0(Y) \setminus \{ S \}$
using again Lemma \ref{lem:W^s_0(Y)}.
\end{proof}

Given $\xi \in \Phi$ with $\phi_1(\xi) > 0$ and $\phi_2(\xi) > 0$,
let $Q$ and $R$ be the upper-most and lower-most points
of $\{ Y, Z, f(Z), f(C), f^2(C) \}$, respectively.
Then let $\Omega$ be the compact filled triangle
with vertices $C$, $Q$, and $R$
(except $\Omega$ is a line segment in the special case $\delta_L = \delta_R = 0$).
In other words, $\Omega$ is the convex hull of 
$Y$, $Z$, $f(Z)$, $f(C)$, $f^2(C)$, and $C$.

There are six cases for the points
that form the vertices of $\Omega$.
These are shown in Fig.~\ref{fig:Omegas} and correspond to the six parts of $\Phi_{\rm trap}$ indicated
in Fig.~\ref{fig:phis}. Fig.~\ref{fig:Omegas} also shows the set $f(\Omega)$.  Notice in each case $f(\Omega)$ has vertices at the images of the points $P$ and $V$ where the boundary of $\Omega$ intersects $x=0$.

\begin{figure}
\begin{center}
\begin{tabular}{cc}
 \includegraphics[scale=0.45]{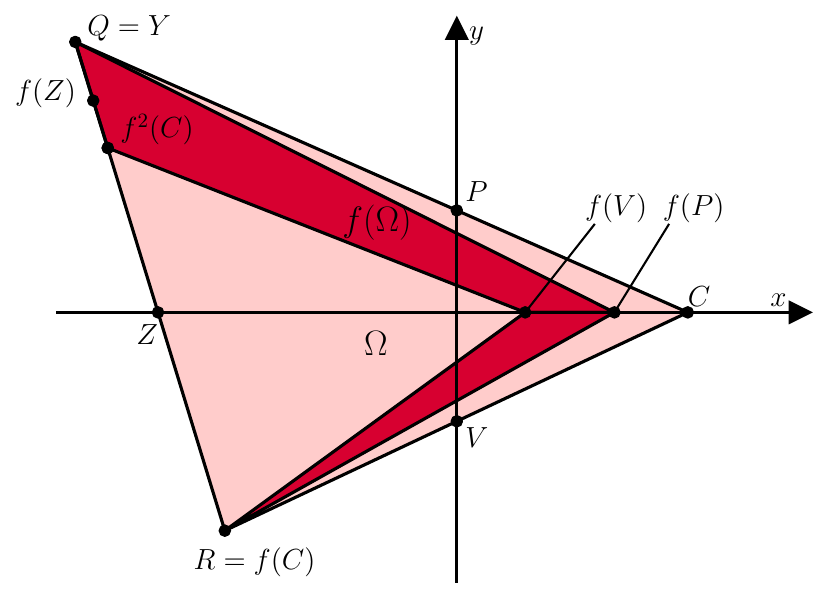} &   \includegraphics[scale=0.45]{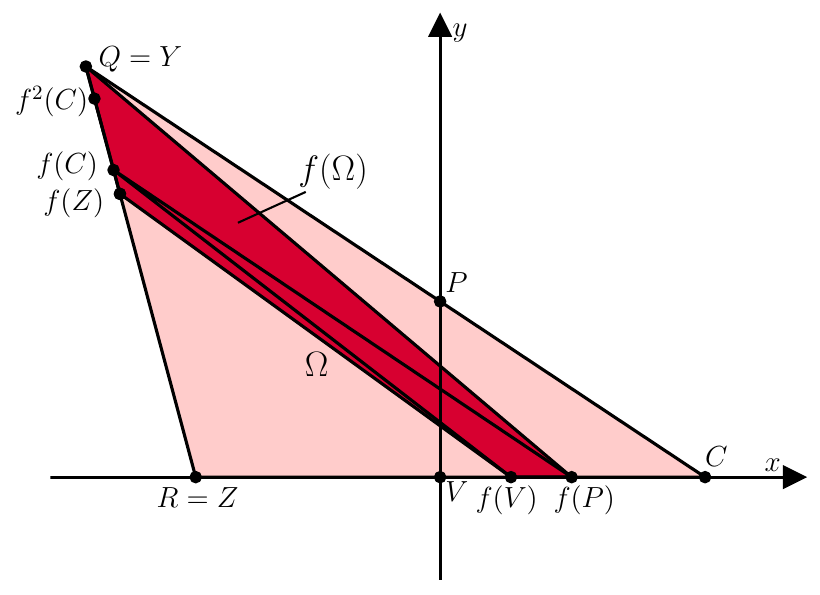} \\
(A) $\delta_L=0.3, \delta_R=0.4$ & (B) $\delta_L=0.4, \delta_R=-0.4$ \\[6pt]
\includegraphics[scale=0.45]{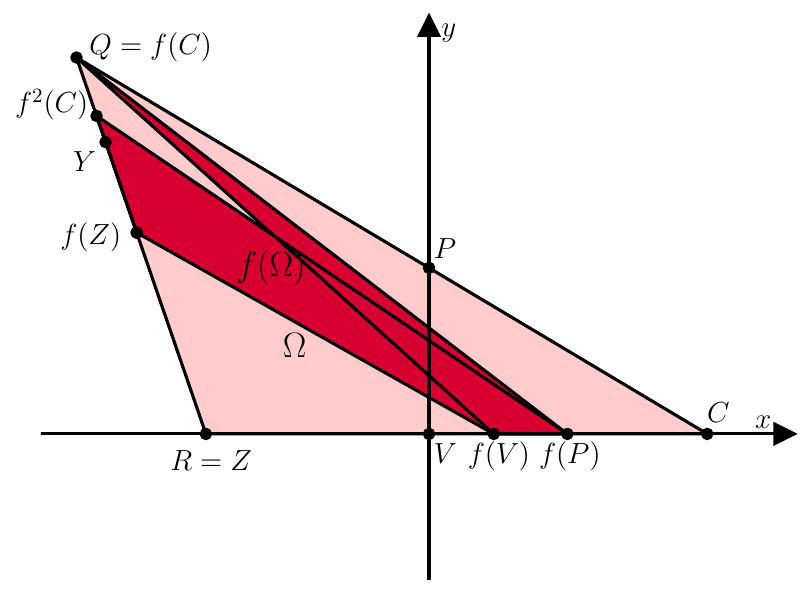} &   \includegraphics[scale=0.45]{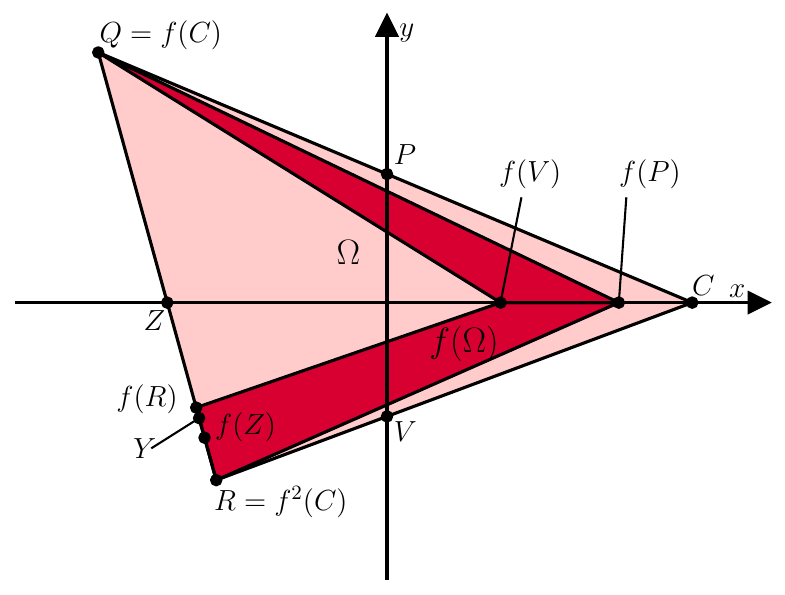} \\
(C) $\delta_L=0.4, \delta_R=-0.6$ & (D) $\delta_L=-0.3, \delta_R=-0.4$  \\[6pt]
\includegraphics[scale=0.45]{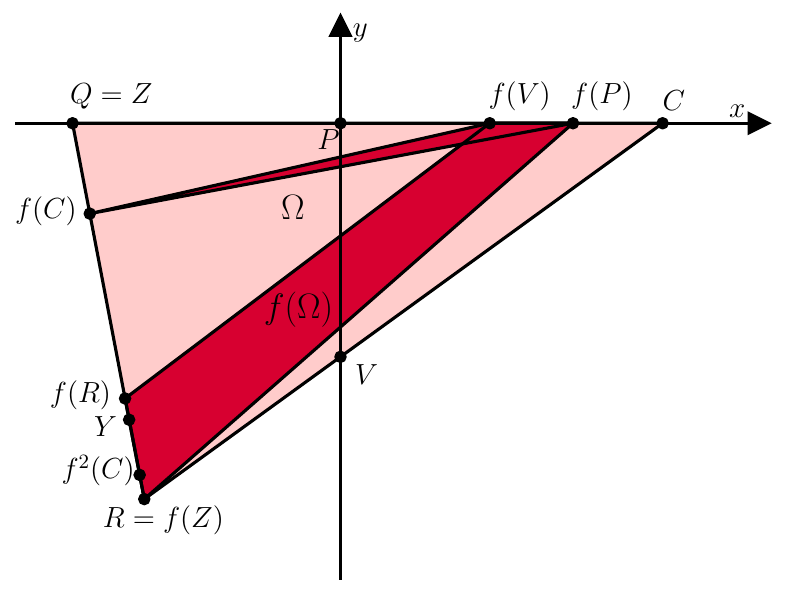} &   \includegraphics[scale=0.45]{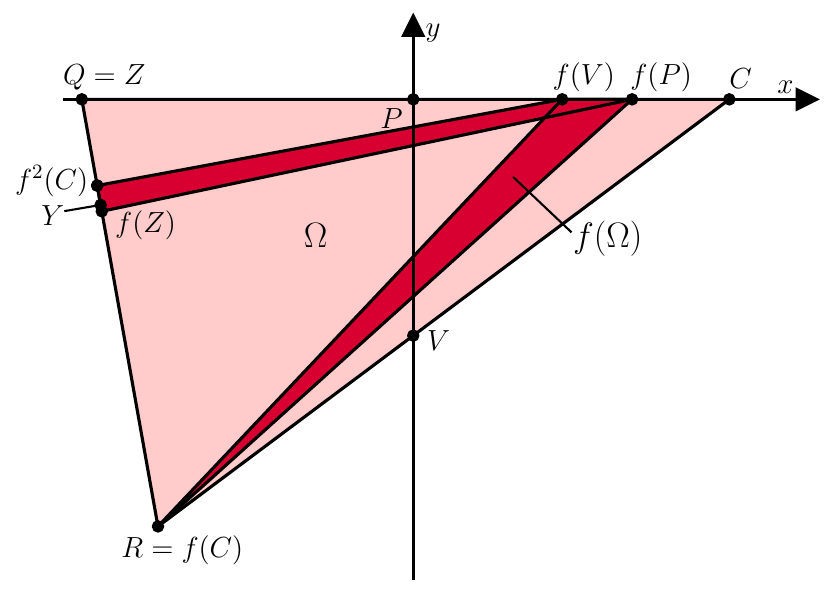} \\
(E) $\delta_L=-0.5, \delta_R=0.1$ & (F) $\delta_L=-0.1, \delta_R=0.4    $  \\[6pt]
\end{tabular}
\end{center}
\caption{
The triangle $\Omega$ and its image under $f$
with $\tau_L=1.6$, $\tau_R=-1.5$, and six different values
of $\delta_L$ and $\delta_R$.
These correspond to parts (A)--(F) of Fig.~\ref{fig:phis}-b for different cases for the two vertices of $\Omega$ in the left half-plane, denoted $Q$ and $R$ (the third vertex is always $C$).
In each case $f(\Omega) \subset \Omega$ by Proposition~\ref{pr:for_inv} because
the parameter combinations belong to $\Omega_{\rm trap}$.
\label{fig:Omegas}}
\end{figure}

\begin{proposition}
\label{pr:for_inv}
Let $\xi \in \Phi_{\rm trap}$.
Then $f(\Omega) \subset \Omega$.
\end{proposition}

\begin{proof}
The proof is long so we break it into three steps.

\myStep{1}{Characterise $f(\Omega)$.}
The vertices $Q$ and $R$ lie in the left half-plane,
while $C$ lies in the right half-plane.
Let $P$ denote the intersection of $\overline{Q C}$
(the line segment from $Q$ to $C$) with $x=0$,
and $V$ denote the intersection of $\overline{R C}$ with $x=0$,
see Fig.~\ref{fig:Omegas}.
From the formula \eqref{eq:C} for $C$,
the $y$-components of these points are given in terms of $Q$ and $R$ by
\begin{align}
P_2
&= \frac{{\lambda_L^u}^3Q_2}{{\lambda_L^u}^3 + [\lambda_L^u - Q_2(1-\lambda_L^u)] \phi_1(\xi)},
\label{eq:P2} \\
V_2
&= \frac{{\lambda_L^u}^3R_2}{{\lambda_L^u}^3 + [\lambda_L^u - R_2(1-\lambda_L^u)] \phi_1(\xi)}.
\label{eq:V2}
\end{align}
So $\Omega$ is the union of
the quadrilateral $\Omega_L$ in the left half-plane
with vertices $P$, $Q$, $R$, and $V$,
and the triangle $\Omega_R$ in the right half-plane
with vertices $C$, $P$, and $V$.
Thus $f(\Omega) = f(\Omega_L) \cup f(\Omega_R)$,
where $f(\Omega_L)$ and $f(\Omega_R)$ are polygons
because each piece of $f$ is affine.
Thus since $\Omega$ is convex, to prove $f(\Omega) \subset \Omega$
it suffices to show that the vertices of
$f(\Omega_L)$ and $f(\Omega_R)$
belong to $\Omega$.
These vertices are the points
$f(C)$, $f(P)$, $f(Q)$, $f(R)$, and $f(V)$.

\myStep{2}{Show $C$, $Q$, and $R$ map to $\Omega$.}
Certainly $f(C) \in \Omega$ by the definition of $\Omega$.
We now show $f(Q), f(R) \in \overline{QR}$
(the left edge of $\Omega$).
If $\delta_L > 0$ then $\lambda_L^s > 0$,
so $f(Q) \in \overline{QY}$
and $f(R) \in \overline{RY}$
so certainly $f(Q), f(R) \in \overline{QR}$.
Also if $\delta_L = 0$, then $f(Q) = f(R) = Y = Z \in \overline{QR}$.
Finally if $\delta_L < 0$, then $f(Q)$ and $f(R)$
lie below $y=0$, so lie below $Z$, and hence below $Q$.
In this case $Q$ is either $f(C)$ or $Z$
(because $Y$, $f(Z)$, and $f^2(C)$ lie below $y = 0$),
thus $f(Q)$ lies on or above $R$, by the definition of $R$.
Also $\lambda_L^s < 0$, thus $f(R)$ lies above $f(Q)$,
and hence above $R$.
Thus in any case $f(Q), f(R) \in \overline{QR}$.

\myStep{3}{Show $P$ and $V$ map to $\Omega$.}
The points $f(P)$ and $f(V)$ lie on $y=0$, specifically
\begin{align}
f(P) &= \begin{bmatrix} P_2 + 1 \\ 0 \end{bmatrix}, &
f(V) &= \begin{bmatrix} V_2 + 1 \\ 0 \end{bmatrix},
\end{align}
where $P_2$ and $V_2$ are given by \eqref{eq:P2} and \eqref{eq:V2}.
Also $V$ lies above $S$, thus $f(V)$ lies to the right of $Z=f(S)$.
Hence it remains for us to show that $f(P)$ lies at or to the left of $C$,
that is $C_1 - (P_2 + 1) \ge 0$.
To do this we consider the various possibilities for $Q$ in turn.
There are three cases: $Q$ is either $Y$, $Z$, or $f(C)$.
This is because $f(Z)$ cannot lie above $Y$, while $f^2(C)$ cannot lie above $Y$ if $\delta_L > 0$
and cannot lie above $Z$ if $\delta_L \le 0$.

{\bf Case 1}: With $Q = Z$ we have $P_2 = 0$.
Thus $C_1 - (P_2 + 1) > 0 = C_1 - 1 > 0$ because $\phi_3(\xi) > 0$.

{\bf Case 2}: With $Q = Y$, by substituting $Y_2$, given by \eqref{eq:Y},
in place of $Q_2$ in \eqref{eq:P2} we obtain
\begin{equation}
C_1 - (P_2 + 1) = \left(1 + \frac{\lambda_L^s{\lambda_L^u}^2}{{\lambda_L^u}^2(1-\lambda_L^s)+\phi_1(\xi)}\right)\frac{\phi_4(\xi)}{\phi_1(\xi)(\lambda_L^u - 1)}.
\nonumber
\end{equation}
This case requires $\delta_L \ge 0$, thus $\lambda_L^s \ge 0$.
Also $C_1 - (P_2 + 1) > 0$
because $\phi_1(\xi)>0$ and $\phi_4(\xi) >0$.

{\bf Case 3}: With $Q = f(C)$ we similarly obtain
\begin{equation}
C_1 - (P_2 + 1) = \left(1 + \frac{\lambda_R^s{\lambda_L^u}^2}{\lambda_L^u(\lambda_L^u-\delta_R)\phi_1(\xi)} \right)\frac{\phi_5(\xi)}{\phi_1(\xi)(\lambda_L^u - 1)}.
\nonumber
\end{equation}
This case requires $\delta_R \le 0$ (so that $f(C)_2 \ge 0$), thus $\lambda_R^s \ge 0$.
Also $\phi_1(\xi)>0$ and $\phi_5(\xi) >0$,
thus $C_1 - (P_2 + 1) > 0$.
\end{proof}


\begin{proposition}
\label{pr:trap}
Let $\xi \in \Phi_{\rm trap}$.
Given $\ep > 0$ define
\begin{align}
C_{\ep} &= C - (\ep,0), \nonumber \\
Q_{\ep} &= Q + \ep^2(C - R), \nonumber \\
R_{\ep} &= R + \ep^2(C - Q), \nonumber
\end{align}
and let $\Omega_\ep$
be the compact filled triangle with vertices
$C_\ep$, $Q_\ep$, and $R_\ep$.
Then $f(\Omega_\ep) \subset {\rm int}(\Omega_\ep)$,
for all sufficiently small $\ep > 0$.
\end{proposition}

The triangle $\Omega_\ep$ is shown in Fig.~\ref{fig:trapping_region} for one combination of parameter values.
It has been defined so that its left edge lies to the right of $W_0^s(Y)$ and is parallel to $W^s_0(Y)$.
Due to the saddle nature of $Y$, this edge shifts further to the right when iterated under $f$.
Also the left edge is an order $\ep^2$ distance from $W_0^s(Y)$
to ensure $f(C_\ep)$ lies to the right of the image of this edge.

\begin{figure}[b!]
\vskip 6pt
\centering
\includegraphics[width=.7\linewidth]{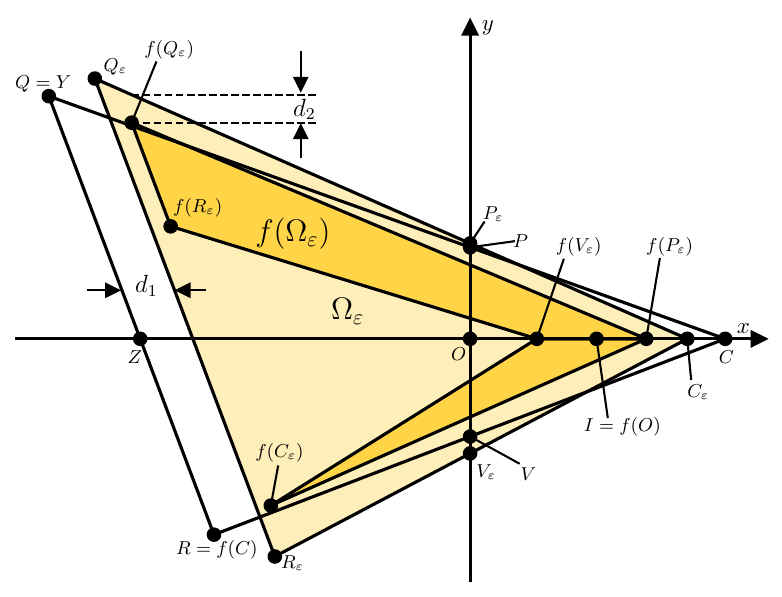}
\caption{
The trapping region $\Omega_\ep$
and its image under $f$.
This figure was produced using
$\xi = \{1.6, 0.3, -1.5, 0.4 \}$, as in Fig.~\ref{fig:Omegas}-A,
and $\ep = 0.3$ which is small enough that $f(\Omega_\ep) \subset {\rm int}(\Omega_\ep)$ in accordance with Proposition~\ref{pr:trap}.
}\label{fig:trapping_region}
\end{figure}

\begin{proof}
\myStep{1}{Characterise $f(\Omega_\ep)$.}
Assume $\ep > 0$ is sufficiently small
that $Q_\ep$ and $R_\ep$ lie to the left of $x=0$
and $C_\ep$ lies to the right of $x=0$.
Then the line segment $\overline{C_\ep Q_\ep}$
intersects $x=0$ at a unique point $P_\ep$,
as does $\overline{C_\ep R_\ep}$ at a point $V_\ep$.
Similar to the previous proof,
it remains for us to show
that $C_\ep$, $P_\ep$, $Q_\ep$, $R_\ep$,
and $V_\ep$ map to the interior of $\Omega_\ep$.

\myStep{2}{Show $C_\ep$ maps to the interior
of $\Omega_\ep$.}
Let $O = (0,0)$ denote the origin
and $I = (1,0)$ be its image under $f$.
Also let $\ell = \overline{C O}$.  This line segment maps under
the right piece of $f$ to the line segment from $f(C_\ep)$ to $I$.
Since $C_\ep \in \ell$ is an order $\ep$ distance from $C$,
its image $f(C_\ep) \in f(\ell)$
is an order $\ep$ distance from $f(C)$, which belongs to $\overline{Q R}$.
Since $\overline{Q_\ep R_\ep}$ is an order $\ep^2$
distance from $\overline{Q R}$,
$f(C_\ep)$ must lie to the right of $\overline{Q_\ep R_\ep}$
for sufficiently small $\ep > 0$.

Also, $f(\ell)$ lies inside the triangle $Q R I$.
Since $C_\ep$ lies to the right of $I$
(because $\phi_3(\xi) > 0$ and assuming $\ep$
is sufficiently small),
this triangle lies below $\overline{Q_\ep C_\ep}$
and above $\overline{R_\ep C_\ep}$.
Thus $f(C_\ep)$ lies below
$\overline{Q_\ep C_\ep}$
and above $\overline{R_\ep C_\ep}$.
Thus $f(C_\ep)$ lies inside all three edges of
$\Omega_\ep$, hence $f(C_\ep) \in {\rm int}(\Omega_\ep)$.

\myStep{3}{Show $P_\ep$ and $V_\ep$ map to the interior of $\Omega_\ep$.}
Let $A$ be any point on $\overline{C Z}$ with $A \ne C, Z$.
Then there exists sufficiently small $\ep > 0$
such that $A \in {\rm int}(\Omega_\ep)$.
We know $f(P)$ and $f(V)$ are located on $\overline{C Z} \setminus \{ C, Z \}$,
because by construction $Z_1 \le f(V)_1 < f(P)_1 < C_1$,
so belong to ${\rm int}(\Omega_\ep)$ for sufficiently small $\ep > 0$.
The same is true for $f(P_\ep)$ and $f(V_\ep)$
because $P_\ep \to P$ and $V_\ep \to V$ as $\ep \to 0$.

\myStep{4}{Show $Q_\ep$ and $R_\ep$ map to the interior of $\Omega_\ep$.}
For brevity we just show $f(Q_\ep) \in {\rm int}(\Omega_\ep)$ ($f(R_\ep) \in {\rm int}(\Omega_\ep)$
can be shown similarly). Let $d_1 > 0$ be the distance that $\overline{Q_\ep R_\ep}$ lies
to the right of $\overline{QR}$, see Fig.~\ref{fig:trapping_region}.
Then $f(Q_{\ep})$ lies a distance $\lambda_L^u d_1$ to the right of
$\overline{QR}$, as $Y$ is a saddle fixed point
with stable direction $\overline{QR}$ and unstable eigenvalue $\lambda_L^u$.
Since $\lambda_L^u>1$, the point $f(Q_{\ep})$ lies to the right of the line $\overline{Q_{\ep}R_{\ep}}$.

If $Q \ne Y$, then $f(Q)$ lies an order $1$ distance below $\overline{CQ}$,
thus $f(Q_\ep)$ lies below $\overline{C_\ep Q_\ep}$
for sufficiently small $\ep > 0$.
Now consider the case $Q=Y$.
As shown in Fig.~\ref{fig:trapping_region} let $d_2$ be the vertical displacement from $f(Q_\ep)$
upwards to the line $\overline{Q_\ep C_\ep}$
(we will show $d_2 > 0$).
By direct calculations $d_2 = \beta \ep^2 + O \left( \ep^3 \right)$ where
\begin{equation}
\beta = \delta_L (C_1 - R_1) (C_1 + (\tau_L - 2) Y_1)
- R_2 (C_1 + (\delta_L - 1) Y_1).
\nonumber
\end{equation}
Let $p = C_1 - D_1$ and notice $p > 0$
by \eqref{eq:phi4origin} because $\phi_4(\xi) > 0$.
Using also $D_1 = \left( 1 - \lambda_L^u \right) Y_1$
by \eqref{eq:Y} and \eqref{eq:D}, we obtain
\begin{equation}
\beta = \delta_L (C_1 - R_1) \left( p - Y_1 \left( 1 - \lambda_L^s \right) \right)
- R_2 \left( p - Y_1 \lambda_L^u \left( 1 - \lambda_L^s \right) \right),
\nonumber
\end{equation}
which is positive by inspection (e.g.~$Y_1 < 0$).
This shows that $d_2 > 0$ for sufficiently small values of $\ep$,
that is $f(Q_\ep)$ lies below the upper edge of $\Omega_\ep$.
By similar calculations one can show that $f(Q_\ep)$ lies above
the lower edge of $\Omega_\ep$,
and therefore $f(Q_\ep) \in {\rm int}(\Omega_\ep)$.
\end{proof}

\section{Invariant expanding cones}
\label{sec:cone}

In this section we first define cones and what it means for them to be invariant and expanding under an arbitrary matrix $A$.  We then focus on the Jacobian matrices $A_L$ and $A_R$ of the normal form (1), because in order for a cone to establish chaos in (1) it needs to be invariant and expanding for both $A_L$ and $A_R$.  We then explicitly construct such a cone for any $\xi \in \Phi_{\rm cone}$ (Proposition 3), and use this to prove Theorem 2.2.

\begin{definition}
A set $C \subset \mathbb{R}^2$ is a {\em cone}
if $\alpha v \in C$ for all $v \in C$ and $\alpha \in \mathbb{R}$.
\label{df:cone}
\end{definition}

\begin{definition}
Given $A \in \mathbb{R}^{2 \times 2}$, a cone $C \subset \mathbb{R}^2$ is
\begin{enumerate}
\item[i)]
{\em invariant} under $A$ if $A v \in C$ for all $v \in C$, and
\item[ii)]
{\em expanding} under $A$ if there exists $c > 1$ such that $\| A v \| \ge c \| v \|$ for all $v \in C$.
\end{enumerate}
\label{df:iec}
\end{definition}

In this paper we use the Euclidean norm $\| v \| = \sqrt{v_1^2 + v_2^2}$ and it suffices to consider cones of the form
\begin{equation}
\Psi_K = \left\{ \alpha \begin{bmatrix} 1 \\ m \end{bmatrix} \,\middle|\, \alpha \in \mathbb{R}, m \in K \right\},
\label{eq:PsiK}
\end{equation}
where $K$ is an interval.
Since $v \mapsto A v$ is a linear map,
to verify invariance and expansion of a cone $\Psi_K$,
it suffices to verify properties (i) and (ii) for vectors
of the form $v = \begin{bmatrix} 1 \\ m \end{bmatrix}$:

\begin{lemma}
If $A v \in \Psi_K$ for all $v = \begin{bmatrix} 1 \\ m \end{bmatrix}$ with $m \in K$,
then $\Psi_K$ is invariant under $A$.
If there exists $c > 1$ such that $\| A v \| \ge c \| v \|$ for all
$v = \begin{bmatrix} 1 \\ m \end{bmatrix}$ with $m \in K$,
then $\Psi_K$ is expanding under $A$.
\label{le:suff}
\end{lemma}

Now we focus on the Jacobian matrices
\begin{equation}
A_J = \begin{bmatrix} \tau_J & 1 \\ -\delta_J & 0 \end{bmatrix},
\nonumber
\end{equation}
where $J \in \{ L, R \}$, of the normal form \eqref{eq:BCNF2}.
The slope of $v = \begin{bmatrix} 1 \\ m \end{bmatrix}$ is $m$
and the slope of $A_J v$ is
\begin{align}
\label{eq:GJ}
G_J(m) = -\frac{\delta_J}{\tau_J+m},
\end{align}
assuming $m \ne -\tau_J$.
Notice
\begin{align}
\label{eq:d_GJ}
\frac{d G_J(m)}{dm} = \frac{\delta_J}{(\tau_J+m)^2},
\end{align}
thus $G_J(m)$ is increasing if $\delta_J > 0$,
decreasing if $\delta_J < 0$,
and flat if $\delta_J = 0$.
In any case, $G_J(m)$ is monotone and so in order to verify invariance
under $A_J$ it suffices to consider the endpoints of $K$:

\begin{lemma}
Let $\tau_J, \delta_J \in \mathbb{R}$ and $K = [a,b]$ be an interval
with $-\tau_J \notin K$.
If $a \le G_J(a) \le b$ and $a \le G_J(b) \le b$
then $\Psi_K$ is invariant under $A_J$.
\label{le:inv}
\end{lemma}

\begin{proof}
Since $-\tau_J \notin K$,
by \eqref{eq:GJ} and \eqref{eq:d_GJ} $G_J(m)$ is continuous and monotone on $K$.
Thus for any $m \in K$, $G_J(m)$ is equal to or lies between the values $G_J(a)$ and $G_J(b)$.
Thus $a \le G_J(m) \le b$.  That is, the slope of $A_J \begin{bmatrix} 1 \\ m \end{bmatrix}$
belongs to $K$, thus $A_J \begin{bmatrix} 1 \\ m \end{bmatrix} \in \Psi_K$.  Hence $\Psi_K$ is invariant under $A_J$ by Lemma \ref{le:suff}.
\end{proof}

Next we introduce the function
\begin{equation}
H_J(m) = \left\| A_J \begin{bmatrix} 1 \\ m \end{bmatrix} \right\|^2
- \left\| \begin{bmatrix} 1 \\ m \end{bmatrix} \right\|^2
= \tau_J^2 + \delta_J^2 - 1 + 2 \tau_J m.
\label{eq:HJ}
\end{equation}
It is easy to show that if $H_J(m) > 0$ for all $m$ in a compact interval $K$,
then $\Psi_K$ is expanding under $A_J$.
Since $H_J(m)$ is a linear function of $m$ it again suffices to consider the endpoints of $K$:

\begin{lemma}
Let $\tau_J, \delta_J \in \mathbb{R}$ and $K = [a,b]$ be an interval.
If $H_J(a) > 0$ and $H_J(b) > 0$ then $\Psi_K$ is expanding under $A_J$.
\label{le:exp}
\end{lemma}

\begin{proof}
Let $h = \min[H_J(a),H_J(b)] > 0$.
By \eqref{eq:HJ}, $H_J(m) \ge h$ for all $m \in K$.
Then for any $m \in K$ the vector $v = \begin{bmatrix} 1 \\ m \end{bmatrix}$ satisfies
\begin{equation}
\| A_J v \|^2 = H_J(m) + \| v \|^2 \ge h + \| v \|^2
= \left( \tfrac{h}{\| v \|^2} + 1 \right) \| v \|^2
\ge \left( \frac{h}{n} + 1 \right) \| v \|^2,
\nonumber
\end{equation}
where $n = \max_{m \in K} (1+m^2)$.
Thus $\Psi_K$ is expanding under $A_J$
(with expansion factor $c = \sqrt{\frac{h}{n} + 1} > 1$)
by Lemma \ref{le:suff}.
\end{proof}

To prove chaos in~\eqref{eq:BCNF2} we need to choose $K = [a,b]$ so that $\Psi_K$ is invariant under $A_L$ and $A_R$.
This favours the interval $K$ being relatively large. However, we want $K$ to be as small as possible in order
to maximise the parameter region over which it is expanding under $A_L$ and $A_R$. This balancing act motivates the following calculations that form the basis of our definition of $K$ given below in Proposition~\ref{pr:cone}.

For each $J \in \{ L, R \}$, the fixed point equation $G_J(m) = m$ is quadratic in $m$.
If $\delta_J \ne 0$ and $\delta_J < \frac{\tau_J^2}{4}$,
then $G_J$ has exactly two fixed points
\begin{align}
q_J &= -\frac{\tau_J}{2}\left(1 - \sqrt{1 - \frac{4\delta_J}{\tau_J^2}} \right), &
r_J &= -\frac{\tau_J}{2}\left(1 + \sqrt{1 - \frac{4\delta_J}{\tau_J^2}} \right).
\label{eq:qJrJ}
\end{align}
\begin{figure}
\begin{center}
\begin{tabular}{cc}
  \includegraphics[scale=0.45]{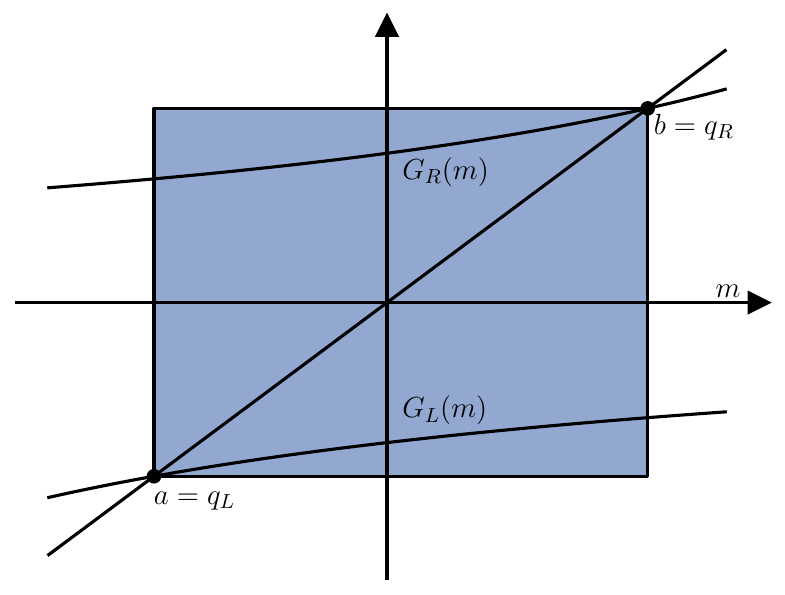} &   \includegraphics[scale=0.45]{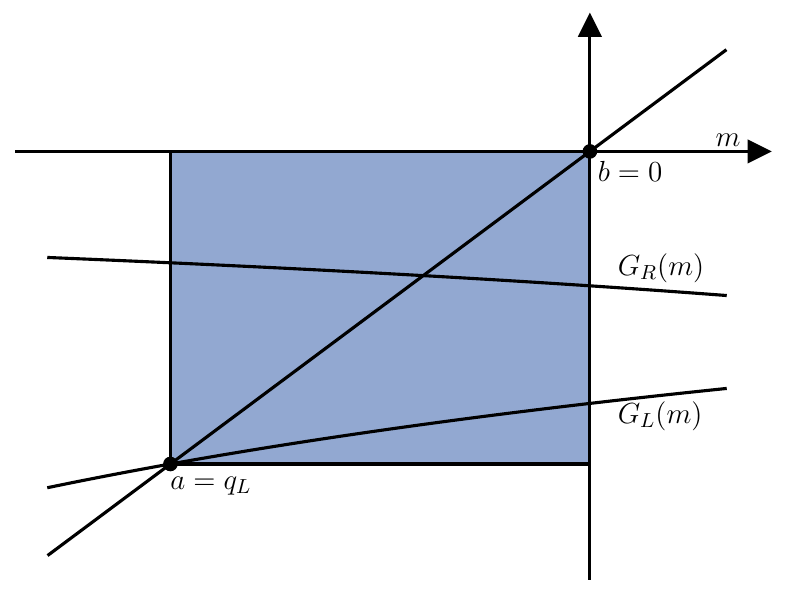} \\
(I) $\delta_L=0.4$, $\delta_R=0.4$. & (II) $\delta_L=0.4$, $\delta_R=-0.2$.\\[6pt]
 \includegraphics[scale=0.45]{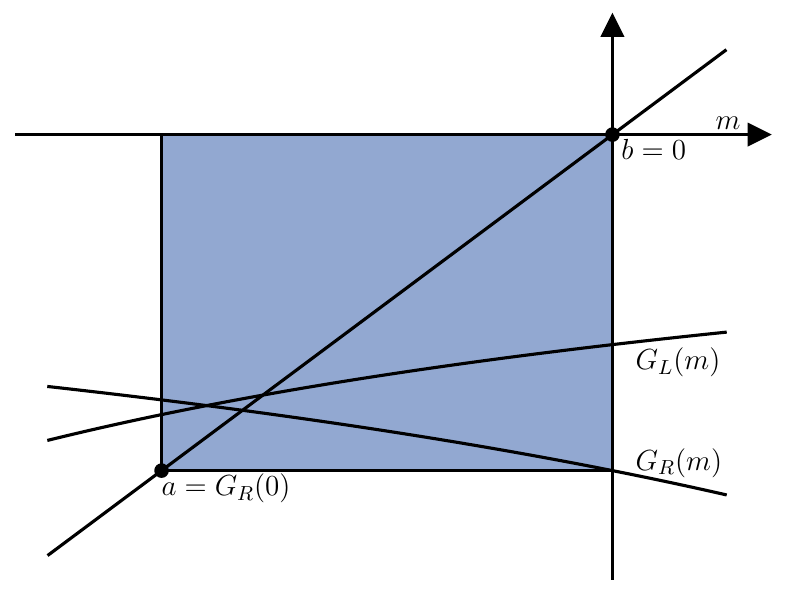} &   \includegraphics[scale=0.45]{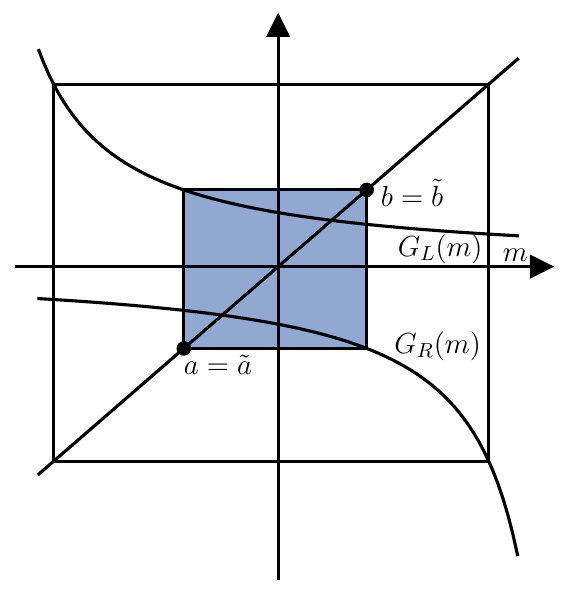} \\
(III) $\delta_L=0.4$, $\delta_R=-0.6$. & (IV) $\delta_L=-0.5$, $\delta_R=-0.5$.\\[6pt]
 \includegraphics[scale=0.45]{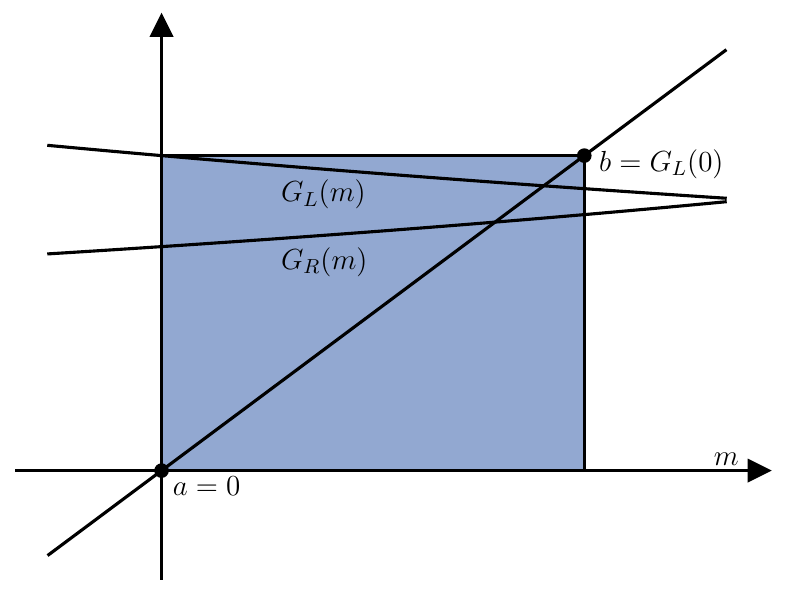} &   \includegraphics[scale=0.45]{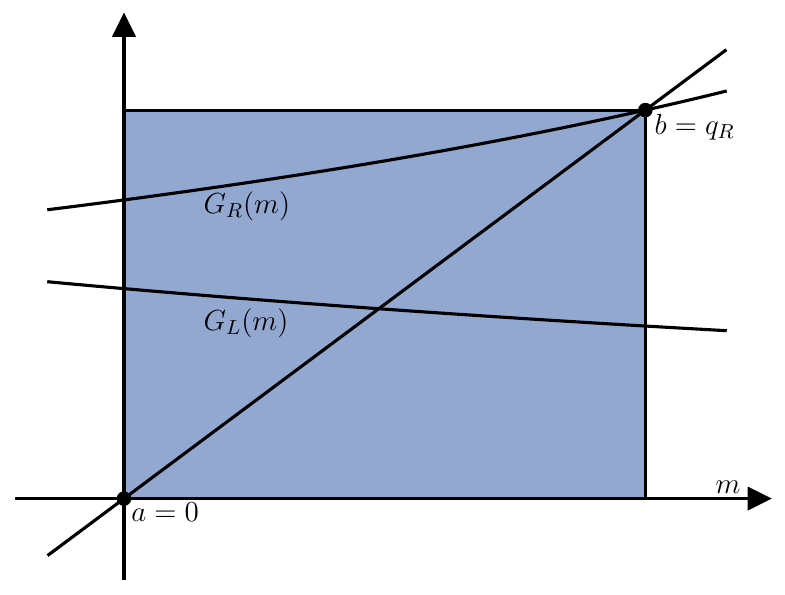} \\
(V) $\delta_L=-0.3$, $\delta_R=0.2$. & (VI) $\delta_L=-0.3$, $\delta_R=0.4$. \\[6pt]
\end{tabular}
\end{center}
\caption{
Cobweb diagrams of the slope maps $G_L(m)$ and $G_R(m)$ \eqref{eq:GJ}
with $\tau_L=1.6$, $\tau_R=-1.5$,
and six different values of $\delta_L$ and $\delta_R$.
These correspond to parts (I)--(VI) of Fig.~\ref{fig:cone_regions}-b.
In each case $K = [a,b]$, defined by~\eqref{eq:a} and \eqref{eq:b},
is forward invariant under $G_L$ and $G_R$ as shown.
This implies $\Psi_K$ is invariant under $A_L$ and $A_R$ (see Lemma~\ref{le:inv}).  In each case $K$ contains the fixed points $q_L$ and $q_R$; in case (IV) $K$ is defined using one of the two period-two solutions.
}
\label{fig:slope_maps}
\end{figure}In order for $\Psi_K$ to be invariant under $A_L$ and $A_R$, we define $K$ so that it contains $q_L$ and $q_R$, see Fig.~\ref{fig:slope_maps}. So the smallest interval we consider is $K = [q_L,q_R]$,.  In the orientation-preserving case ($\delta_L, \delta_R > 0$), this interval indeed gives invariance, as shown in \cite{GlSi21}. In the orientation-reversing case ($\delta_L, \delta_R < 0$),
invariance requires that $K$ contains a period-two solution.
The equation $(G_R \circ G_L)(m) = m$ is quadratic in $m$ with discriminant
\begin{equation}
\theta_1(\xi) = \left( \delta_L + \delta_R - \tau_L \tau_R \right)^2
- 4 \delta_L \delta_R \,,
\nonumber
\end{equation}
repeating \eqref{eq:theta1}.
So if $\theta_1(\xi) > 0$ there are two period-two solutions, Fig.~\ref{fig:slope_maps}-(IV).
Of these the inner-most solution is $\{ \tilde{a},\tilde{b} \}$, where
\begin{align}
\tilde{a} &= \frac{\delta_L -\delta_R-\tau_L\tau_R - \sqrt{\theta_1(\xi)}}{2\tau_R}, &
\tilde{b} &= \frac{\delta_R -\delta_L-\tau_L\tau_R - \sqrt{\theta_1(\xi)}}{2\tau_L},
\label{eq:tildeatildebagain}
\end{align}
satisfying $G_L(\tilde{a}) = \tilde{b}$ and $G_R(\tilde{b}) = \tilde{a}$.
So in this case the smallest interval we can take is $K = [\tilde{a},\tilde{b}]$, as
used by Misiurewicz \cite{Mi80}.

In the non-invertible cases
the slope maps $G_L$ and $G_R$ are either both non-negative or both non-positive, see again Fig.~\ref{fig:slope_maps}.
Thus a simple and effective choice for one endpoint of $K$ is $m=0$.
In this case the smallest interval leading to invariance uses also one of $q_L$, $q_R$, or the image of $m=0$ under $G_L$ or $G_R$:
\begin{align}
G_L(0) &= -\frac{\delta_L}{\tau_L}, &
G_R(0) &= -\frac{\delta_R}{\tau_R}.
\label{eq:GL0GR0}
\end{align}

Proposition~\ref{pr:cone} shows that all cases
can be accommodated by simply defining $a$ and $b$ as the minimum and maximum of all points suggested above.
Recall $\Phi_{\rm cone}$ was defined in \S\ref{sec:main} as the set of all $\xi \in \Phi$
for which $\theta_1(\xi)$, $\theta_2(\xi)$, and $\theta_3(\xi)$ are positive.
The condition $\theta_1(\xi) > 0$ ensures $\tilde{a}$ and $\tilde{b}$ are well-defined,
while, if $a$ and $b$ are given by \eqref{eq:a} and \eqref{eq:b},
\begin{align}
\theta_2(\xi) &= \tau_L^2 + \delta_L^2 - 1 + 2 \tau_L a = H_L(a), \label{eq:theta2simplified} \\
\theta_3(\xi) &= \tau_R^2 + \delta_R^2 - 1 + 2 \tau_R b = H_R(b), \label{eq:theta3simplified}
\end{align}
and $\theta_2(\xi) > 0$ and $\theta_3(\xi) > 0$ ensure $\Psi_K$ is invariant and expanding.

\begin{proposition}
Let $\xi \in \Phi_{\rm cone}$ and $K = [a,b]$ where
\begin{align}
a &= \min \left[ 0, -\tfrac{\delta_R}{\tau_R}, q_L, \tilde{a} \right], \label{eq:a} \\
b &= \max \left[ 0, -\tfrac{\delta_L}{\tau_L}, q_R, \tilde{b} \right]. \label{eq:b}
\end{align}
Then $\Psi_K$ is invariant and expanding under $A_L$ and $A_R$.
\label{pr:cone}
\end{proposition}

To prove Proposition \ref{pr:cone} we first establish three lemmas. The first of these provides bounds on the other fixed points, $r_L$ and $r_R$, of $G_L$ and $G_R$.

\begin{lemma}
\label{le:rL}
Let $\xi \in \Phi$.
Then $r_L < \frac{1 - \delta_L^2 - \tau_L^2}{2 \tau_L}$
and $r_R > \frac{1 - \delta_R^2 - \tau_R^2}{2 \tau_R}$.
\end{lemma}

\begin{proof}
We have $\delta_L+1< \tau_L$, hence $(\delta_L+1)^2 < \tau_L^2$, and so $(\delta_L-1)^2 < \tau_L^2 - 4\delta_L$. By multiplying the last two inequalities together we obtain $(\delta_L^2 - 1)^2 < \tau_L^2(\tau_L^2 - 4\delta_L)$, so $\delta_L^2-1 < \tau_L\sqrt{\tau_L^2-4\delta_L}$ which rearranges to
$r_L < \frac{1 - \delta_L^2 - \tau_L^2}{2 \tau_L}$
using \eqref{eq:qJrJ}.
The result for $r_R$ follows similarly.
\end{proof}

\begin{lemma}
With the assumptions of Proposition \ref{pr:cone},
$a = \tilde{a}$ if and only if $\delta_L \le 0$ and $\delta_R \le 0$;
similarly $b = \tilde{b}$ if and only if $\delta_L \le 0$ and $\delta_R \le 0$.
\label{le:B}
\end{lemma}

\begin{proof}
First suppose $\delta_L \le 0$ and $\delta_R \le 0$.
Then $\tau_L \tau_R \delta_L \ge 0$ (also $\theta_1(\xi) > 0$ by assumption),
so we can use \eqref{eq:tildeatildebagain} to obtain
\begin{equation}
2 \tau_R (\tilde{a} + \tau_L) = -\sqrt{\theta_1(\xi)} - \sqrt{\theta_1(\xi) + 4 \tau_L \tau_R \delta_L} < 0.
\nonumber
\end{equation}
Thus $\tilde{a} > -\tau_L$ (because $\tau_R < 0$).
Also $\tilde{b} < -\tau_R$ by a similar argument.

Notice $\delta_L \le 0$ implies $G_L(m) \ge 0$ for all $m > -\tau_L$,
so $q_L \ge 0$, $G_L(0) = -\frac{\delta_L}{\tau_L} \ge 0$, and $G_L(\tilde{a}) = \tilde{b} \ge 0$.
Similarly $G_R(m) \le 0$ for all $m < -\tau_R$,
so $q_R \le 0$, $G_R(0) = -\frac{\delta_R}{\tau_R} \le 0$, and $G_R(\tilde{b}) = \tilde{a} \le 0$.
Also $G_L(m)$ is non-increasing, so $\tilde{a} \le 0$ implies $G_L(\tilde{a}) \ge G_L(0)$.
Thus $\tilde{b} \ge -\frac{\delta_L}{\tau_L} \ge 0 \ge q_R$, so $b = \tilde{b}$.
Similarly $\tilde{a} \le -\frac{\delta_R}{\tau_R} \le 0 \le q_L$, so $a = \tilde{a}$.

Conversely suppose $a = \tilde{a}$.
Then $\tilde{a} = G_R(\tilde{b}) \le 0$, so $\delta_R \le 0$.
Thus $G_R(m)$ is non-increasing, so $\tilde{a} = G_R(\tilde{b}) \le G_R(0)$ implies $\tilde{b} \ge 0$.
Thus $\tilde{b} = G_L(\tilde{a}) \ge 0$, so $\delta_L \le 0$, as required
(also $b = \tilde{b}$ implies $\delta_L \le 0$ and $\delta_R \le 0$ in a similar fashion).
\end{proof}

\begin{lemma}
With the assumptions of Proposition \ref{pr:cone}, $-\tau_L < a$ and $b < -\tau_R$.
\label{le:C}
\end{lemma}

\begin{proof}
For brevity we just show $-\tau_L < a$ ($b < -\tau_R$ can be shown similarly). Using $\theta_2(\xi) > 0$ and Lemma \ref{le:rL} we obtain
\begin{equation}
a > \frac{1 - \delta_L^2 - \tau_L^2}{2 \tau_L} > r_L \,.
\label{eq:aLargerThanrL}
\end{equation}
Thus if $\delta_L \ge 0$ then
\begin{equation}
r_L + \tau_L = \frac{\tau_L}{2} \left( 1 - \sqrt{1 - \frac{4 \delta_L}{\tau_L^2}} \right) \ge 0,
\nonumber
\end{equation}
and so $a > -\tau_L$.
Now suppose $\delta_L < 0$.
If $\delta_R \ge 0$ then $a = 0 > -\tau_L$,
while if $\delta_R < 0$ then $a = \tilde{a}$ (by Lemma \ref{le:B})
and $-\tau_L < \tilde{a}$ as shown in the proof of Lemma \ref{le:B}.
\end{proof}

\begin{proof}[Proof of Proposition \ref{pr:cone}]
We first show $G_L(a) \ge a$ and $G_L(b) \ge a$.
If $\delta_L \le 0$ then $G_L(m) \ge 0$ for all $m \in K$
(using Lemma \ref{le:C}), so certainly $G_L(a) \ge a$ and $G_L(b) \ge a$. Now suppose $\delta_L > 0$. In this case $G_L(m) \ge m$ for all $r_L \le m \le q_L$ (i.e.~at and between the fixed points of $G_L$).
Observe $r_L < a \le q_L$ by \eqref{eq:aLargerThanrL}
and the definition of $a$, thus $G_L(a) \ge a$.
Also $G_L(m)$ is increasing thus $G_L(b) > G_L(a) \ge a$.

Next we show $G_L(a) \le b$ and $G_L(b) \le b$.
If $\delta_L \ge 0$ then we have $G_L(m) \le 0$ for all $m \in K$ (using Lemma \ref{le:C}), so certainly $G_L(a) \le b$ and $G_L(b) \le b$.
Now suppose $\delta_L < 0$. If $\delta_R \le 0$ then Lemma~\ref{le:B} implies $a = \tilde{a}$ and $b = \tilde{b} = G_L(\tilde{a})$, so $b=G_L(a)$, if $\delta_R > 0$ then $a =0$ and $b\ge G_L(0)$, so $b \ge G_L(a)$.
Also $G_L(m)$ is decreasing thus $G_L(b) < G_L(a) \le b$.

Now from Lemma~\ref{le:inv} we can conclude that $\Phi_K$ is invariant under $A_L$. Invariance under $A_R$ can be proved in a similar fashion.

Next we prove expansion. By \eqref{eq:theta2simplified}, $\theta_2(\xi) > 0$
implies $H_L(a) > 0$.
Also
\begin{equation}
H_L(b) = \tau_L^2 + \delta_L^2 - 1 + 2 \tau_L b = H_L(a) + 2 \tau_L (b - a)
\end{equation}
is positive because $\tau_L > 0$ and $b \ge a$.
Thus $\Psi_K$ is expanding for $A_L$ by Lemma \ref{le:exp}.
By a similar argument $\Psi_K$ is also expanding for $A_R$.
\end{proof}

\begin{proof}[Proof of Theorem~\ref{thm:attractor}]
Choose any $\xi \in \Phi_{\rm trap} \cap \Phi_{\rm cone}$.
By Proposition \ref{pr:trap} there exists a trapping region $\Omega_\ep$ for $f$.
Then $\bigcap_{n \ge 0} f^n(\Omega_\ep)$ is an attracting set
and contains a topological attractor $\Lambda$.
By Proposition \ref{pr:cone} there exists a non-empty cone $\Psi_K$ that is invariant and expanding
for both $A_L$ and $A_R$ with some expansion factor $c > 1$.

Choose any $z \in \Lambda$ and let $v \in \Psi_K$ be non-zero.
The Lyapunov exponent $\lambda(z,v)$ for $z$ in the direction $v$
is the limiting rate of separation of the forward orbits of $z$ and $z + \Delta v$ for arbitrarily small $\Delta > 0$ \cite{Vi14}.
If the forward orbit of $z$ does not intersect the switching manifold then
the derivative of the $n^{\rm th}$ iterate of $z$ under $f$ is well-defined for all $n \ge 1$ and
\begin{equation}
\lambda(z,v) = \limsup_{n \to \infty} \frac{1}{n} \ln \left( \left\| {\rm D} f^n(x) v \right\| \right).
\label{eq:Lyap}
\end{equation}
Observe
\begin{equation}
{\rm D} f^n(x) = {\rm D} f \left( f^{n-1}(x) \right) \cdots {\rm D} f \left( f(x) \right) {\rm D} f(x),
\nonumber
\end{equation}
where each of the $n$ matrices on the right-hand side is either $A_L$ or $A_R$.
By the invariance and expansion of $\Psi_K$, $\left\| {\rm D} f^n(x) v \right\| \ge c^n \| v \|$ for all $n$, so $\lambda(z,v) \ge \ln(c) > 0$.
If instead the forward orbit of $z$ intersects the switching manifold,
$\lambda(z,v)$ can similarly be evaluated and bounded using
one-sided directional derivatives because $f$ is piecewise-linear, see \cite{Si22e} for details.
\end{proof}

\section{Further remarks on the parameter regions $\Phi_{\rm trap}$ and $\Phi_{\rm cone}$.}
\label{sec:topologies}

In \S\ref{sec:main} we described two-dimensional cross-sections of $\Phi_{\rm trap}$ and $\Phi_{\rm cone}$ defined by fixing the values of $\tau_L > 0$ and $\tau_R < 0$.
For the most part larger values of $\tau_L$ and $|\tau_R|$ yield smaller cross-sections of $\Phi_{\rm trap}$ and larger cross-sections of $\Phi_{\rm cone}$, see  Fig.~\ref{fig:phis} and Fig.~\ref{fig:cone_regions}.
This is because with larger values of $\tau_L$ and $|\tau_R|$ the map is more strongly expanding,
hence less amenable for the existence of a trapping region
but more amenable for the existence of an invariant expanding cone.

\begin{figure}[h]
\vskip 6pt
\centering
\includegraphics[width=0.7\linewidth]{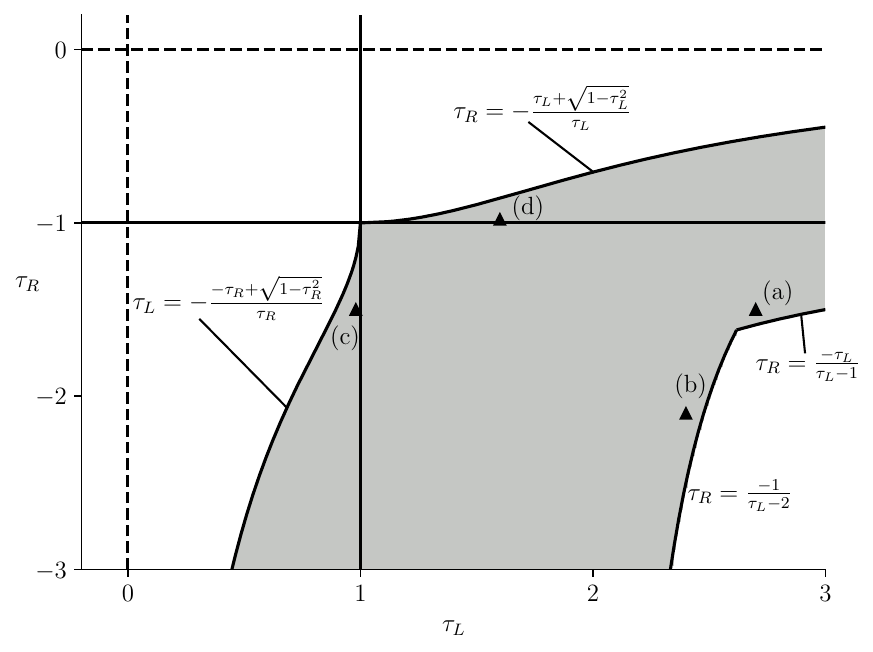}
\caption{
The shaded region shows where cross-sections of $\Phi_{\rm trap}$ and $\Phi_{\rm cone}$, defined by fixing the values of $\tau_L$ and $\tau_R$, are both non-empty. As we exit the shaded region through the lower-right boundaries the cross-section of $\Phi_{\rm trap}$ vanishes;
as we exit through the upper and left boundaries the cross-section of $\Phi_{\rm cone}$ vanishes.
}\label{fig:tltr}
\end{figure}

Fig.~\ref{fig:tltr} shows critical curves in the $(\tau_L,\tau_R)$-plane where the cross-sections vanish entirely.
To explain this figure we treat the critical curves one by one.  First consider $(\tau_L,\tau_R)$ at a point just above the critical curve $\tau_R = \frac{-\tau_L}{\tau_L-1}$.  Here the $\Phi_{\rm trap}$ cross-section has three vertices, $P^{(1)}$, $P^{(2)}$, and $P^{(3)}$, as shown in Fig.~\ref{fig:Phi_IE}-a.  It is a simple exercise to show that as parameters are varied each vertex reaches the origin $(\delta_L,\delta_R) = (0,0)$ on the critical curve.  For instance the upper vertex is where $\phi_3(\xi) = 0$ and $\phi_4(\xi) = 0$ intersect at $P^{(1)} = \left(0, \frac{(\tau_R - (\tau_R+1)\tau_L)\tau_L}{1 - \tau_L} \right)$ and solving $P^{(1)}_2 = 0$ gives $\tau_R = \frac{-\tau_L}{\tau_L - 1}$.  Thus here the $\Phi_{\rm trap}$ cross-section contracts to a point and vanishes.

With instead $(\tau_L,\tau_R)$ at a point just to the left of $\tau_R = \frac{-1}{\tau_L - 2}$, the $\Phi_{\rm trap}$ cross-section has three vertices at different points $P^{(4)}$, $P^{(5)}$, and $P^{(6)}$, as shown in Fig.~\ref{fig:Phi_IE}-b.  Explicit calculations reveal that each vertex reaches the corner $(\delta_L,\delta_R) = (\tau_L+1,\tau_R-1)$ when $\tau_R = \frac{-1}{\tau_R - 2}$.  For instance, $P^{(4)} = \left(\frac{(\tau_L\tau_R - \tau_R+1)(\tau_R-1)}{\tau_R^2}, \tau_R-1 \right)$, and solving $P^{(4)}_1 = \tau_L + 1$ gives $\tau_R = \frac{-1}{\tau_R - 2}$.  Thus here the $\Phi_{\rm trap}$ cross-section again vanishes.

\begin{figure}[h]
\begin{center}
\begin{tabular}{cc}
  \includegraphics[scale=0.41]{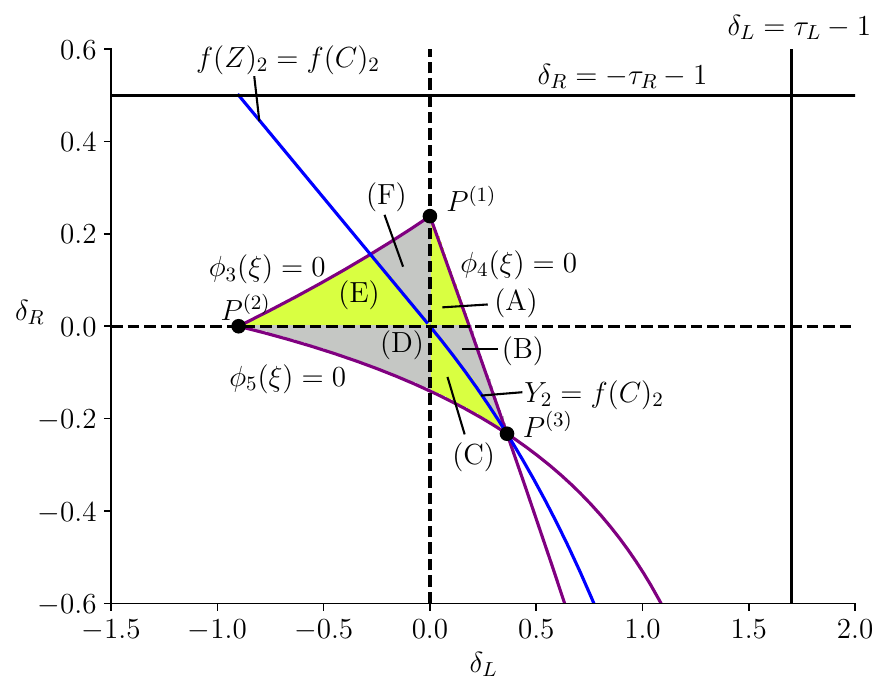} &   \includegraphics[scale=0.41]{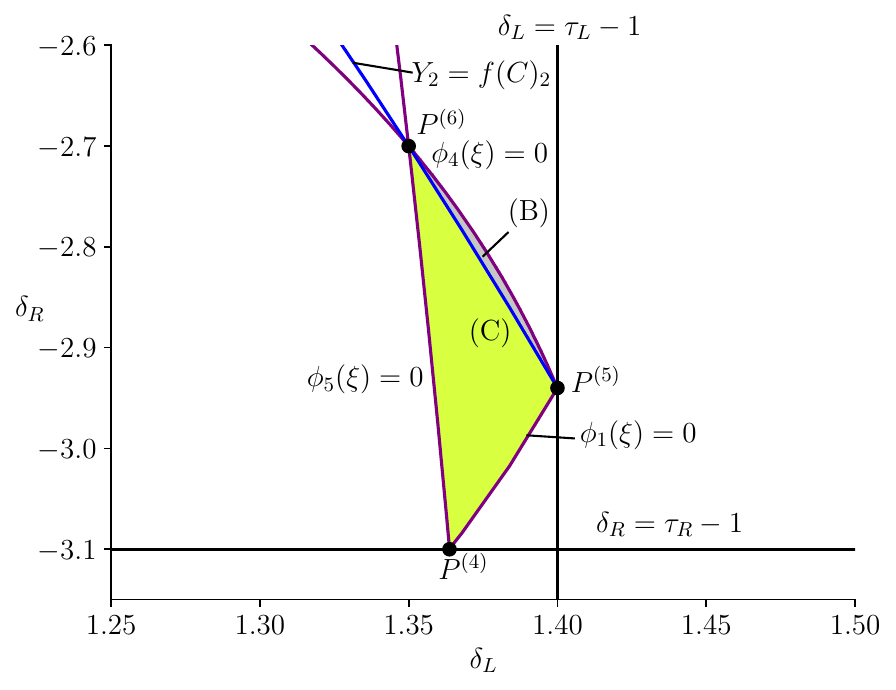} \\
(a) $\tau_L=2.7, \tau_R=-1.5$ & (b) $\tau_L=2.4, \tau_R=-2.1$ \\[6pt]
 \includegraphics[scale=0.41]{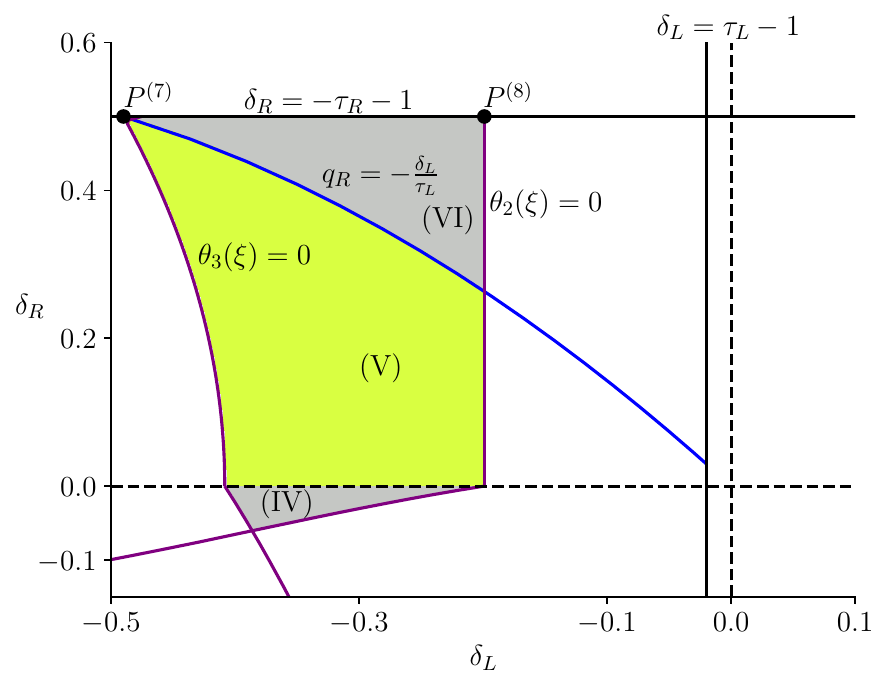} &   \includegraphics[scale=0.41]{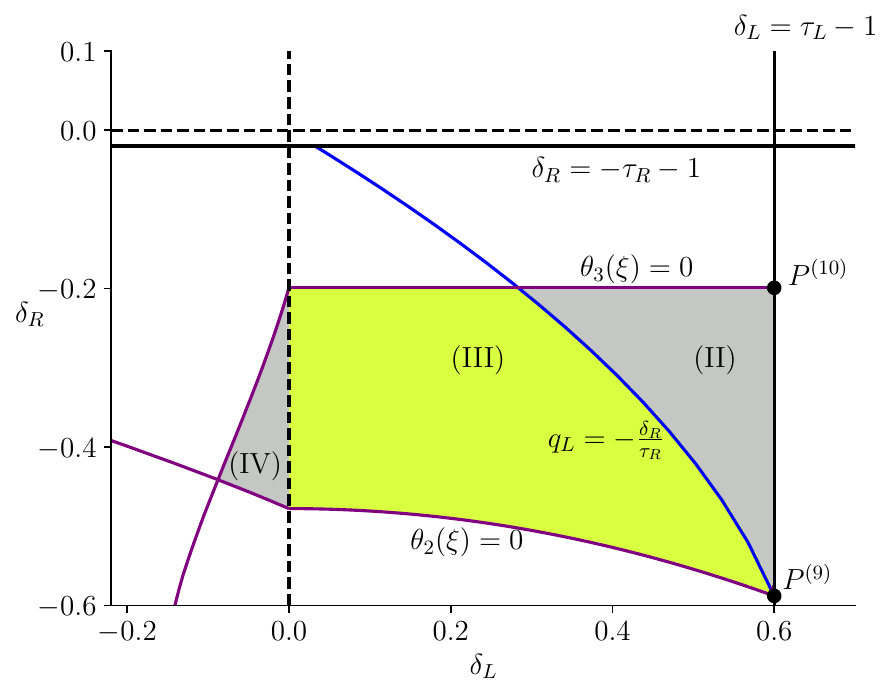} \\
(c) $\tau_L=0.98, \tau_R=-1.5$ & (d) $\tau_L=1.6, \tau_R=-0.98$ \\[6pt]
\end{tabular}
\end{center}
\caption{Further sample cross-sections of $\Phi_{\rm trap}$ (panels (a) and (b)) and $\Phi_{\rm cone}$ (panels (c) and (d)).  These use the parameter points $(\tau_L,\tau_R)$ shown in Fig.~\ref{fig:tltr}.}\label{fig:Phi_IE}
\end{figure}

With $\tau_L$ just less than $1$ and $\tau_R < -1$ the $\Phi_{\rm cone}$ cross-section appears as in Fig.~\ref{fig:Phi_IE}-c.  As parameters are varied the $\Phi_{\rm cone}$ cross-section vanishes when the vertices $P^{(7)} = \big( \tau_L(\tau_R+1), -(\tau_R+1) \big)$ and
$P^{(8)} = \big( -\sqrt{1-\tau_L^2}, -(\tau_R+1) \big)$ coincide. Solving $P^{(7)}_1 = P^{(8)}_1$ yields the critical curve $\tau_R = -\frac{\tau_L+\sqrt{1-\tau_L^2}}{\tau_L}$ of Fig.~\ref{fig:tltr}. Similarly with $\tau_L > 1$ and $\tau_R$ just greater than $-1$ the $\Phi_{\rm cone}$ cross-section appears as in Fig.~\ref{fig:Phi_IE}-d.  The cross-section vanishes when $P^{(9)}$ and $P^{(10)}$ coincide on $\tau_L = -\frac{-\tau_R+\sqrt{1-\tau_R^2}}{\tau_R}$.

\begin{figure}[h]
\vskip 6pt
\centering
\includegraphics[width=0.7\linewidth]{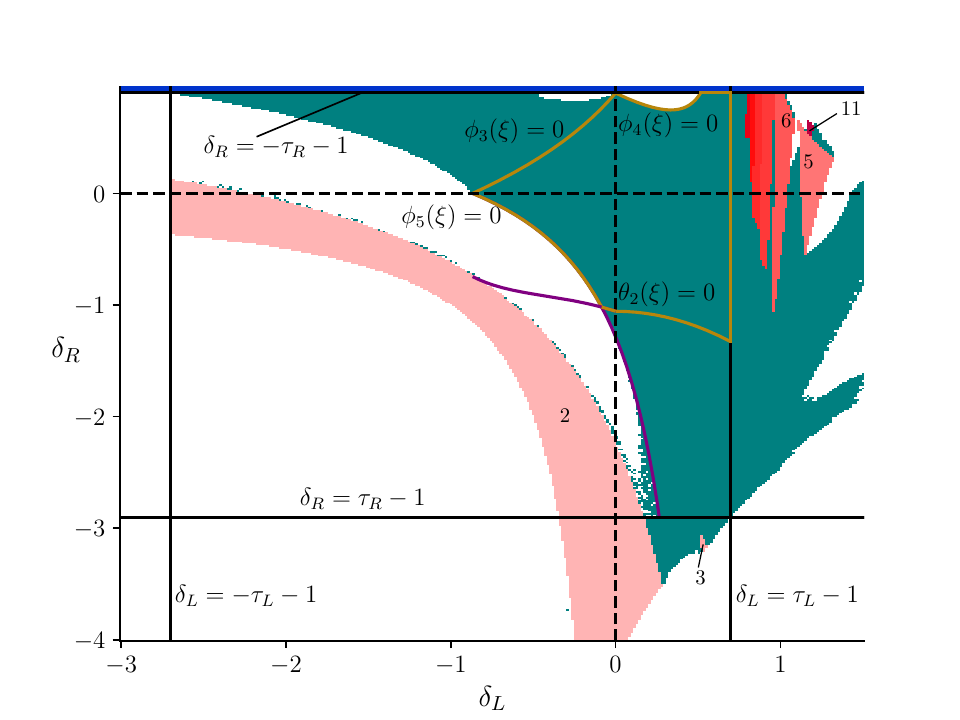}
\caption{A two-parameter bifurcation diagram of (1) with $\tau_L = 1.7$ and $\tau_R = -1.9$.  The boundary of $\Phi_{\rm trap} \cap \Phi_{\rm cone}$ is shown in yellow.  The coloured regions show the result of a numerical simulation (white: no attractor; green: chaotic attractor; other colours: periodic attractor).  Some periods are indicated (e.g.~in the large pink region the map has a stable period-$2$ solution.  In the blue strip at the top the fixed point $X$ is stable.
}\label{fig:numerics}
\end{figure}

The geometry and topology of cross-sections of $\Phi_{\rm trap} \cap \Phi_{\rm cone}$ admit many possibilities
and a complete analysis is beyond the scope of this paper.
Here we examine one example, Fig.~\ref{fig:numerics}.
Here the cross-section of $\Phi_{\rm trap} \cap \Phi_{\rm cone}$
is bounded by the following curves (going anticlockwise):
$\delta_L = \tau_L - 1$,
$\delta_R = -\tau_R - 1$,
$\phi_4(\xi) = 0$,
$\phi_3(\xi) = 0$,
$\phi_5(\xi) = 0$,
and $\theta_2(\xi) = 0$ (which has a kink at $\delta_L = 0$).
The first two of these curves are boundaries of our overall parameter region $\Phi$, the next three curves are boundaries of $\Phi_{\rm trap}$,
and the last curve is a boundary of $\Phi_{\rm cone}$.

Fig.~\ref{fig:numerics} also shows the result of a simple numerical simulation
to investigate the nature of the attractor.
For each point in a $300 \times 300$ equispaced  grid of $(\delta_L,\delta_R)$ values,
we computed $10^7$ iterates of the forward orbit using a random initial condition.
Green points are where an estimate of the maximal Lyapunov exponent was positive, white points are where the orbit appeared to diverge (its norm exceeded $10^4$), and other points are where there exists a stable periodic solution (determined by solving for periodic solutions exactly).

Note the numerical simulation gives an imperfect picture.
For example, we believe no attractor exists immediately to the left of the heteroclinic bifurcation $\phi_5(\xi) = 0$, yet some green points are present here (with small $\delta_L > 0$ and $-3 < \delta_R < -2$) because orbits often experience a long transient before diverging
and $10^7$ iterations are insufficient to detect this.

Nevertheless, the numerics effectively highlights the fact that three of the boundaries of $\Phi_{\rm trap} \cap \Phi_{\rm cone}$
are bifurcations where the chaotic attractor is destroyed, so these boundaries cannot be improved upon.
As we cross $\delta_R = -\tau_R - 1$ the fixed point $X$ becomes stable, $\phi_4(\xi) = 0$ is a homoclinic bifurcation where the attractor is destroyed \cite{BaYo98}, and $\phi_5(\xi) = 0$ is a heteroclinic bifurcation where the attractor is destroyed.
Elsewhere the attractor is destroyed in a variety of other global bifurcations.
Above $\phi_3(\xi) = 0$ and to the right of $\delta_L = \tau_L - 1$ the trapping region construction of \S\ref{sec:for_inv_trap} fails.
An alternate construction that partially deals with this is described in~\cite{Si22e}.
Below $\theta_2(\xi) = 0$ the cone $\Psi_K$ of \S\ref{sec:cone} fails to be expanding.
For some parameter combinations below $\theta_2(\xi) = 0$ it is possible to construct an invariant expanding cone, and hence verify the presence of chaos, by using an induced map \cite{GlSi22b}.

\section{Discussion}
\label{sec:dis}

We have extended the constructive proof of~\cite{GlSi21} for robust chaos in the border-collision normal form to the orientation-reversing and non-invertible settings. Specifically we identified an open parameter region $\Phi_{\rm trap}$, where the map has a trapping region, and an open parameter region $\Phi_{\rm cone}$, where it has an invariant expanding cone, see Figs.~\ref{fig:phis} and \ref{fig:cone_regions}. Throughout $\Phi_{\rm trap} \cap \Phi_{\rm cone}$ the map has an attractor with a positive Lyapunov exponent, Theorem~\ref{thm:attractor}. 

In Fig.~\ref{fig:numerics} we considered a typical two-dimensional cross-section of parameter space.  Numerical results showed robust chaos throughout $\Phi_{\rm trap} \cap \Phi_{\rm cone}$, corroborating Theorem 2.2 as expected.  The robust chaos terminates at three boundaries of $\Phi_{\rm trap} \cap \Phi_{\rm cone}$: $\delta_R = -\tau_R-1$, $\phi_4(\xi) = 0$, and $\phi_5(\xi) = 0$, and appears to persist beyond the other boundaries.  We expect our construction could be adapted to verify robust chaos beyond these boundaries, and already this has been achieved in some cases \cite{GlSi22b,Si22e}.

Notice in Fig.~\ref{fig:numerics} the cross-section of $\Phi_{\rm trap} \cap \Phi_{\rm cone}$ includes a neighbourhood of $(\delta_L,\delta_R)$.  This is the case for many values of $\tau_L$ and $\tau_R$ and is a significant achievement of this paper because it shows robust chaos is not lost as we cross from the orientation-preserving setting to the orientation-reversing and non-invertible settings.  Thus the presence of robust chaos is not dependent on the global topological properties of the map.  Moreover, this provides a path for robust chaos to be demonstrated in higher dimensional maps. The $n$-dimensional border-collision normal form \cite{Si16} can have two saddle fixed points: $X$ with an eigenvalue $\lambda_R^u < -1$, and $Y$ with an eigenvalue $\lambda_L^u > 1$.  If all other eigenvalues associated with $X$ and $Y$ are sufficiently small in absolute value (which in two dimensions means $(\delta_L,\delta_R)$ is sufficiently close to $(0,0)$), and with appropriate constraints on the values of $\lambda_R^u$ and $\lambda_L^u$, we believe the map must have a chaotic attractor, and this is a tantalising avenue for future research.

\section{Acknowledgments}
This work was supported by Marsden Fund contract MAU1809, managed by Royal Society Te Ap\={a}rangi.


\bibliographystyle{plain}
\bibliography{main}

\end{document}